\documentclass[a4paper,UKenglish,cleveref, autoref, thm-restate,]{lipics-v2021}

\usepackage{mathtools}
\usepackage{stmaryrd}
\usepackage{xcolor}
\usepackage{tikz}

\newcommand{\NN}{\mathbb{N}}
\newcommand{\NNpos}{\NN_{\scriptscriptstyle \geq 1}}

\newcommand{\Structure}[1]{\ensuremath{\mathcal{#1}}}

\newcommand{\CC}{\Structure{C}}

\newcommand{\X}{\Structure{X}}

\newcommand{\Algorithm}[1]{\mathscr{#1}}

\newcommand{\bigO}{\mathcal{O}}
\newcommand{\onot}{\bigO}

\newcommand{\smid}{\,;}

\newcommand{\neighb}[3]{\ensuremath{N_{#1}^{#2}(#3)}} \newcommand{\neighbr}[2]{\ensuremath{\neighb{r}{#1}{#2}}} \newcommand{\Neighb}[3]{\ensuremath{\mathcal{N}_{#1}^{#2}(#3)}} \newcommand{\Neighbr}[2]{\ensuremath{\mathcal{N}_{r}^{#1}(#2)}} 

\newcommand{\type}[3]{\ensuremath{\textup{tp}_{#1}(#2,#3)}} \newcommand{\ltype}[4]{\ensuremath{\textup{ltp}_{#1,#2}(#3,#4)}} \newcommand{\Types}[3]{\ensuremath{\textup{Tp}[#1,#2,#3]}}

\newcommand{\abs}[1]{\left\lvert#1\right\rvert}

\DeclareMathOperator{\dist}{dist}

\DeclareMathOperator{\err}{err}

\renewcommand{\leq}{\leqslant}
\renewcommand{\geq}{\geqslant}
\renewcommand{\le}{\leq}
\renewcommand{\ge}{\geq}

\renewcommand{\phi}{\varphi}
\renewcommand{\epsilon}{\varepsilon}

\newtheorem{fact}[theorem]{Fact}
\newsavebox{\mybox}
\newenvironment{problem}[1]{\vspace*{1ex}\begin{lrbox}{\mybox}\begin{minipage}{0.9\linewidth}#1\setlength{\leftmargini}{4pt}\setlength{\leftmarginii}{20pt}\begin{description}\setlength{\labelwidth}{0pt}}{\end{description}\end{minipage}\end{lrbox}\fbox{\usebox{\mybox}}\vspace*{1ex}}

\usepackage[ruled]{algorithm}
\usepackage[noend]{algorithmic}
\newcommand{\algorithmicreject}{\textbf{reject}}
\makeatletter
\newcommand{\REJECT}{\ALC@it\algorithmicreject{} \ }
\makeatother
\newcommand{\algorithmicbreak}{\textbf{break}}
\makeatletter
\newcommand{\BREAK}{\ALC@it\algorithmicbreak{} \ }
\makeatother

\newcommand{\set}[1]{\ensuremath{\{#1\}}}

\newcommand{\ComplexityClassFont}[1]{\ensuremath{\textup{\textsf{#1}}}}

\newcommand{\FPT}{\ensuremath{\ComplexityClassFont{FPT}}}
\newcommand{\AWstar}{\ensuremath{\ComplexityClassFont{AW}[*]}}

\newcommand{\LogicFont}[1]{\ComplexityClassFont{#1}}
\newcommand{\FO}{\LogicFont{FO}}

\newcommand{\FOmc}{\textup{\sc FO-Mc}}

\definecolor{rwth-blue}{cmyk}{1,.5,0,0}\colorlet{rwth-lblue}{rwth-blue!50}\colorlet{rwth-llblue}{rwth-blue!25}
\definecolor{rwth-violet}{cmyk}{.6,.6,0,0}\colorlet{rwth-lviolet}{rwth-violet!50}\colorlet{rwth-llviolet}{rwth-violet!25}
\definecolor{rwth-purple}{cmyk}{.7,1,.35,.15}\colorlet{rwth-lpurple}{rwth-purple!50}\colorlet{rwth-llpurple}{rwth-purple!25}
\definecolor{rwth-carmine}{cmyk}{.25,1,.7,.2}\colorlet{rwth-lcarmine}{rwth-carmine!50}\colorlet{rwth-llcarmine}{rwth-carmine!25}
\definecolor{rwth-red}{cmyk}{.15,1,1,0}\colorlet{rwth-lred}{rwth-red!50}\colorlet{rwth-llred}{rwth-red!25}
\definecolor{rwth-magenta}{cmyk}{0,1,.25,0}\colorlet{rwth-lmagenta}{rwth-magenta!50}\colorlet{rwth-llmagenta}{rwth-magenta!25}
\definecolor{rwth-orange}{cmyk}{0,.4,1,0}\colorlet{rwth-lorange}{rwth-orange!50}\colorlet{rwth-llorange}{rwth-orange!25}
\definecolor{rwth-yellow}{cmyk}{0,0,1,0}\colorlet{rwth-lyellow}{rwth-yellow!50}\colorlet{rwth-llyellow}{rwth-yellow!25}
\definecolor{rwth-grass}{cmyk}{.35,0,1,0}\colorlet{rwth-lgrass}{rwth-grass!50}\colorlet{rwth-llgrass}{rwth-grass!25}
\definecolor{rwth-green}{cmyk}{.7,0,1,0}\colorlet{rwth-lgreen}{rwth-green!50}\colorlet{rwth-llgreen}{rwth-green!25}
\definecolor{rwth-cyan}{cmyk}{1,0,.4,0}\colorlet{rwth-lcyan}{rwth-cyan!50}\colorlet{rwth-llcyan}{rwth-cyan!25}
\definecolor{rwth-teal}{cmyk}{1,.3,.5,.3}\colorlet{rwth-lteal}{rwth-teal!50}\colorlet{rwth-llteal}{rwth-teal!25}
\definecolor{rwth-gold}{cmyk}{.35,.46,.7,.35}
\definecolor{rwth-silver}{cmyk}{.39,.31,.32,.14}

\newcommand{\bigmid}{\mathrel{\big|}}

\newcommand{\CH}{{\mathcal H}}

\newcommand{\FOERM}{\textup{\sc FO-Erm}}
\newcommand{\RFOERM}[1]{\textup{\sc $(#1)$-FO-Erm}}

\title{On the Parameterized Complexity of\\ Learning First-Order Logic}

\titlerunning{On the Parameterized Complexity of
  Learning First-Order Logic} 

\author{Steffen van Bergerem}{RWTH Aachen University,
Germany}{vanbergerem@informatik.rwth-aachen.de}{https://orcid.org/0000-0002-5212-8992}{}
\author{Martin Grohe}{RWTH Aachen University,
Germany}{grohe@informatik.rwth-aachen.de}{https://orcid.org/0000-0002-0292-9142}{}
\author{Martin Ritzert}{RWTH Aachen University,
Germany}{ritzert@informatik.rwth-aachen.de}{https://orcid.org/0000-0002-5322-3684}{}

\authorrunning{S. van Bergerem, M. Grohe, and M. Ritzert}

\Copyright{Steffen van Bergerem, Martin Grohe, and Martin Ritzert}



\category{} 

\relatedversion{}

\EventEditors{John Q. Open and Joan R. Access}
\EventNoEds{2}
\EventLongTitle{42nd Conference on Very Important Topics (CVIT 2016)}
\EventShortTitle{CVIT 2016}
\EventAcronym{CVIT}
\EventYear{2016}
\EventDate{December 24--27, 2016}
\EventLocation{Little Whinging, United Kingdom}
\EventLogo{}
\SeriesVolume{42}
\ArticleNo{23}

\begin{document}

\maketitle

\begin{abstract}
  We analyse the complexity of learning first-order queries
  in a model-theoretic framework for supervised learning introduced by
  (Grohe and Turán, TOCS 2004). Previous research on the complexity of
  learning in this framework focussed on the question of when learning
  is possible in time sublinear in the background structure.

  Here we study the parameterized complexity of the learning problem.
  We have two main results. The first is a hardness result,
  showing that learning first-order queries is at least as hard as
  the corresponding model-checking problem, which implies that on
  general structures it is
  hard for the parameterized complexity class \AWstar{}.
  Our second main contribution is a fixed-parameter tractable
  agnostic PAC learning algorithm for first-order queries over sparse
  relational data (more precisely, over nowhere dense background structures).
\end{abstract}

\section{Introduction}\label{sec:intro}

We study the complexity of learning first-order queries from
examples. This is a Boolean classification problem where hypotheses
are specified by formulas of first-order logic over finite structures.
In Boolean classification problems, the goal is to learn an unknown
Boolean function $f^*:\X\to\{0,1\}$, the \emph{target function},
defined on an \emph{instance space} $\X$, from a sequence
$(x_1,\lambda_1),\ldots,(x_m,\lambda_m)\in\X\times\{0,1\}$ of
\emph{labelled examples}, where (in the simplest setting) for all $i$
the label $\lambda_i$ is the target value $f^*(x_i)$. Given the
labelled examples, the learner computes a hypothesis $h:\X\to\{0,1\}$
that is supposed to be close to the target function $f^*$.  To make
this problem feasible at all, one usually assumes that the hypothesis
is from some restricted \emph{hypothesis class}~$\mathcal
H$. If the target function is also from this class, we call the learning
problem \emph{realisable}, but we do not insist on this.
We mainly frame our results in Valiant's
\emph{probably approximately correct (PAC) learning} model
\cite{Valiant_PAC} and the more general \emph{agnostic PAC learning}
model \cite{Haussler}. In the latter, we do not require a fixed target
function and instead just assume that there is an unknown probability
distribution $\mathcal D$ on $\X\times\{0,1\}$. Then the learner's
goal is to find a hypothesis $h$ in the hypothesis class that
minimises the probability
$\Pr_{(x,\lambda)\sim\mathcal D}\big(h(x)\neq \lambda\big)$ of being wrong on
a random (unseen) example. This probability is known as the
\emph{generalisation error} (or \emph{risk}) of the hypothesis. It is
known that, information theoretically, PAC learning and agnostic PAC
learning are possible if and only if the hypothesis class $\mathcal H$
has bounded VC dimension \cite{bluehrhau+89,vapche71}. It is also
known that if this is the case, then we can cast the algorithmic
problem as a minimisation problem known as \emph{empirical risk
  minimisation}: find the hypothesis $h\in\mathcal H$ that minimises
the \emph{training error} (a.k.a. \emph{empirical risk}), that is, the
fraction of falsely classified training examples. An
approximate minimisation with an additive error $\epsilon$ is
sufficient for an (agnostic) PAC learning algorithm.

After this very brief review of some necessary definitions and facts
from algorithmic learning theory (for more background, see
\cref{sec:learningProblem} and
\cite{KearnsVazirani,Shalev-Shwartz:2014:UML:2621980}), let us now
describe the logical setting from
\cite{GroheTuran_Learnability,GroheRitzert_FO} that we consider
here. In this setting, we want to learn a first-order query from
positive and negative examples---tuples that are in the query answer
and tuples that are not---over a relational database instance $D$,
called the \emph{background structure} in
\cite{GroheTuran_Learnability,GroheRitzert_FO}. Formally, our
hypotheses are specified by first-order formulas
$\phi(\bar x\smid \bar y)$ together with tuples $\bar w$ of elements
of $D$. We can think of these elements as constants used in the
query. Listing them explicitly as ``parameters'' of the query, rather
than implicitly as part of the formula $\phi$, allows for a cleaner
complexity analysis. The formula $\phi(\bar x\smid\bar y)$ together
with the parameter tuple $\bar w$ specifies a Boolean function
$h_{\phi,\bar w}$ defined by $h_{\phi,\bar w}(\bar v)=1$ if
$D\models\phi(\bar v\smid \bar w)$ and $h_{\phi,\bar w}(\bar v)=0$
otherwise.

Previous work in this logical learning framework focussed on learning
algorithms running in time sublinear in the database $D$ (see the
``Related Work'' section below). Here, we analyse the parameterized
complexity.

\subparagraph*{Our Contributions}
The algorithmic problem we study is the empirical risk minimisation
problem $\FOERM$ in the logical framework: we are given a database instance $D$
and a sequence $\Lambda$ of training examples of the form
$(\bar v,\lambda)$, where $\bar v$ is a $k$-tuple of elements of $D$
and $\lambda\in\{0,1\}$, and the goal is find a hypothesis of the form
$h_{\phi,\bar w}$ for a first-order formula $\phi(\bar x\smid\bar y)$ and a
tuple $\bar w$ that minimises the fraction of examples
$(\bar v,\lambda)$ from $\Lambda$ where
$h_{\phi,\bar w}(\bar v)\neq \lambda$. We restrict the hypothesis
space by giving a bound $q$ on the quantifier rank of $\phi$ and a
bound $\ell$ on the length of $\bar w$.  Furthermore, we allow for an
additive error $\epsilon>0$ in the minimisation.
Provided the sequence
$\Lambda$ of training examples is sufficiently long (logarithmic in
the size of $D$ or linear in the VC dimension of the hypothesis
space), this empirical risk minimisation problem is equivalent to PAC
learning; see the discussion in \cref{sec:learningProblem}.
In addition to the additive approximation error $\epsilon$, we also
allow for a relaxation in the quantifier rank and parameter number of the
output hypothesis. We only require them to be bounded in terms of
functions $L,Q$ of $k,\ell,q$ (where we always assume
$L\ge \ell$ and $Q\ge q$). The resulting version of the problem is denoted by
\RFOERM{L,Q}. Note that this relaxation strengthens our hardness
result (\cref{theo:hardness}) and weakens the tractability
result (\cref{theo:tractability}).

In our parameterized complexity analysis, we take
$k,\ell,q,\frac{1}{\epsilon}$ as parameters; the input size depends on
the size $n$ of $D$ and the size $m$ of $\Lambda$. Note that \FOERM\ and \RFOERM{L,Q} trivially have
polynomial-time data complexity, or in terms of parameterized complexity
theory, belong to the class XP. The main question
we study here is whether they are
\emph{fixed-parameter tractable}, that is, solvable in time
$f(k,\ell,q,1/\epsilon)\cdot(n+m)^{\bigO(1)}$ for some function~$f$.

Our first result is a hardness result, which states that our learning
problem is at least as hard as the model-checking problem for
first-order logic under suitable parameterized reductions. Since the
model-checking problem is known to be complete for the parameterized
complexity class \AWstar\ (see~\cite{flugro06}), we can state this result as
follows.

\begin{restatable}{theorem}{hardness}
\label{theo:hardness}
\RFOERM{L,Q} is hard for the parameterized complexity
class \AWstar\ under parameterized Turing reductions (for all $L,Q$).
\end{restatable}

While the theorem may be unsurprising, its proof is
surprisingly hard. Using a combinatorial trick based on Ramsey's
Theorem, we can compute an equivalence relation with a bounded number
of classes that refines first-order equivalence on the vertices of a graph.
We can do this using a polynomial number of oracle calls to \RFOERM{L,Q}. With
this equivalence relation, we can implement a recursive
fixed-parameter tractable model-checking algorithm.

In view of the hardness theorem, we look at restricted classes of
database instances where the model-checking problem is fixed-parameter
tractable. Without loss of generality, we restrict our attention to
vertex-coloured graphs, because arbitrary relational structures can
easily be encoded as graphs.

It is relatively easy to see that for classes with a fixed-parameter
tractable model-checking problem satisfying mild closure conditions,
the \emph{unary} version of the learning problem (i.e.\ where
\mbox{$k=1$}) is reducible to the model-checking problem and hence
fixed-parameter tractable as well. Unfortunately, the simple reduction
argument breaks down for $k\ge 2$. Therefore, we consider specific
graph classes.
One of the most general family of graph
classes with a tractable first-order model-checking problem is the
family of \emph{effectively nowhere dense graph classes}.
Our second main result states that on those the ERM problem
is fixed-parameter tractable.

\begin{theorem}\label{theo:tractability}
  For all nowhere dense graph classes $\CC$, there are functions $L,Q$ such
  that the restriction of \RFOERM{L,Q} to input graphs from $\CC$ is
  fixed-parameter tractable.
\end{theorem}

The proof of this theorem heavily depends on the characterisation of
nowhere dense classes in terms of the so-called \emph{splitter game}
\cite{grokresie17}. Interestingly, the parameters $\bar w$ in the
hypothesis $h_{\phi,\bar w}$ computed by our algorithm
are essentially the vertices the splitter chooses in her winning
strategy for the game.

\subparagraph*{Related Work}
The model-theoretic learning framework we study here was introduced in
\cite{GroheTuran_Learnability}, where the authors proved
information-theoretic learnability results for both first-order and
monadic second-order logic obtained by restricting the background
structures, for example, to be planar or of bounded tree width. The
algorithmic question was first studied in \cite{GroheRitzert_FO},
where it was proved that on graphs of maximum degree $d$, empirical
risk minimisation for first-order definable hypotheses is possible in
time polynomial in $d$ and the number $m$ of labelled examples the
algorithm receives as input, independently of the size of the
background structure $A$.  This was generalised to first-order logic
with counting \cite{vanBergerem} and with weight aggregation
\cite{vanBergeremSchweikardt}. All these results are mainly
interesting in structures of small, say, polylogarithmic degree,
because there they yield learning algorithms running in time sublinear
in the size of the background structure. It was shown in
\cite{GroheLoedingRitzert_MSO,vanBergerem} that sublinear-time learning is no
longer possible if the degree is unrestricted. To address this issue,
in \cite{GroheLoedingRitzert_MSO} it was proposed to introduce a
preprocessing phase where (before seeing any labelled examples) the
background structure is converted to some data structure that supports
sublinear-time learning later. This model was applied to monadic
second-order logic on strings \cite{GroheLoedingRitzert_MSO} and
trees~\cite{GrienenbergerRitzert_Trees}.

The framework is related to, but different from the framework of
\emph{inductive logic programming} (e.g.~\cite{cohpag95,mug91,mugder94}), which may be viewed as
the classical logic-learning framework.

In the database literature, there are various approaches to learning
queries
\cite{aboangpap+13,DBLP:conf/icdt/Barcelo017,DBLP:conf/edbt/BonifatiCL15,bonciusta16,tencate_et_al:LIPIcs.ICDT.2021.9,DBLP:journals/ml/Haussler89,DBLP:conf/alt/Hirata00,DBLP:conf/icdt/StaworkoW12,DBLP:conf/cp/Willard10}. Many
of these are concerned with active learning scenarios (most recently \cite{tencate_et_al:LIPIcs.ICDT.2021.9}), whereas we are
in a statistical learning setting. Furthermore, most of these results
are concerned with conjunctive queries (or queries outside of the
relational database model), whereas our focus is on full first-order logic.
A related problem that has also
received some attention in the database literature is learning schema
mappings
\cite{DBLP:journals/tods/AlexeCKT11,DBLP:journals/jacm/GottlobS10,DBLP:journals/tods/CateDK13,DBLP:conf/pods/CateK0T18}.

Related logical learning frameworks have also been studied in formal
verification (e.g.~\cite{DBLP:conf/tacas/ChampionC0S18,DBLP:conf/cav/0001LMN14,DBLP:journals/pacmpl/EzudheenND0M18,DBLP:conf/pldi/PadhiSM16,DBLP:conf/pldi/ZhuMJ18}).

 \section{Preliminaries}\label{sec:Prelim}

We use \(\NN{}\) and $\NNpos$ to denote the non-negative and positive
integers and use the shorthand $[m] = \set{1,\dots,m}$.
All variables in this paper denote integers except for \(\epsilon, \delta\)
and their variations such as \(\epsilon^*\) and~\(\delta_0\).

For the ease of presentation, we will state and prove all our results
for (coloured) graphs instead of relational structures (or relational databases). All results
can easily be extended to arbitrary relational structures, either
directly, or indirectly by coding relational structures as graphs.

We consider undirected and simple graphs that may be vertex-coloured.
Formally, we view a graph as a relational structure
$G=\big(V(G),E(G),P_1(G),\ldots,P_\ell(G)\big)$ of some relational vocabulary
\(\tau=\{E,P_1,\ldots,P_\ell\}\) where \(E\) is binary and the \(P_j\)
are unary.  We always assume the edge relation \(E(G)\) to be
symmetric and irreflexive.  If we want to specify the vocabulary
explicitly, we call such a graph a \emph{\(\tau\)-coloured graph}.
For \(\tau' \supseteq \tau\), a \emph{\(\tau'\)-expansion} of a
\(\tau\)-coloured graph \(G\) is a \(\tau'\)-coloured graph \(G'\)
such that \(V(G')=V(G)\), \(E(G')=E(G)\), and \(P(G')=P(G)\) for all
\(P \in \tau\).
The \emph{order} of a graph is the number of its vertices \(|V(G)|\).
A graph \(G'\subseteq G\) is a subgraph of \(G\) if \(V(G')\subseteq V(G) \),
\(E(G')\subseteq E(G)\), and \(P(G') \subseteq P(G)\) for all \(P\in \tau\).
Given a set \(S \subseteq V(G)\), the \emph{induced subgraph} \(G[S]\) is
defined as the graph with nodes \(S\) and all relations restricted to the
nodes in \(S\).
The \emph{distance} \(\dist(u,v)\) between
two nodes \(u,v\in V(G)\) is the length of the shortest path between
$u$ and $v$.
The distance between a node \(u\in V(G)\) and a vector \(\bar v \in \big(V(G)\big)^k\)
is defined as \(\dist(u,\bar v) = \min_{v\in \bar v} \big(\dist(u,v)\big)\).
Similarly, we define the distance between two vectors
\(\bar u \in \big(V(G)\big)^k\) and \(\bar v\in \big(V(G)\big)^\ell\) as
\(\dist(\bar u,\bar v) = \min_{u \in \bar u,\,v\in \bar v}\! \big(\!\dist(u,v)\big)\).
The \(r\)-\emph{neighbourhood} of a vector \(\bar u \in \big(V(G)\big)^k\) is
the set \(\neighbr{G}{\bar{u}} = \set{v \in V(G) \mid \dist(v,\bar u) \leq r}\)
of nodes up to a distance of \(r\) from \(\bar u\).
The \(r\)-neighbourhood of a set \(S\subseteq V(G)\)
is \(\neighbr{G}{S} = \set{u \in V(G) \mid s\in S,\, \dist(u,s) \leq r}\).
The graph \(\Neighbr{G}{\bar{u}} = G[\neighbr{G}{\bar{u}}]\)
induced from the neighbourhood of \(\bar u\) is called the (induced)
\(r\)-neighbourhood graph.
The following lemma is an easy consequence of the so-called Vitali
Covering Lemma.

\begin{restatable}{lemma}{vitali}
\label{lem:vitali}
  Let \(G\) be a graph, \(X \subseteq V(G)\), and \(r \ge 1\).
  Then there is a set \(Z \subseteq X\) and an \(R=3^ir\), where \(0 \leq i \leq {\abs{X}-1}\), such that
  \begin{romanenumerate}
    \item \(\neighb{R}{G}{z} \cap \neighb{R}{G}{z'} = \emptyset\)
    for all distinct \(z, z' \in Z\) and
    \item \(\neighbr{G}{X} \subseteq \neighb{R}{G}{Z}\).
  \end{romanenumerate}
\end{restatable}

\begin{proof}
  For \(i \geq 0\), let \(R_i \coloneqq 3^ir\).
  We inductively construct a sequence \(Z_0 \supset Z_1 \supset \dots \supset Z_p\)
  of subsets of \(X\) such that for every \(i\) we have \(\neighbr{G}{X} \subseteq \neighb{R_i}{G}{Z_i}\)
  and for \(p\) we additionally have \(\neighb{R_p}{G}{z} \cap \neighb{R_p}{G}{z'} = \emptyset\) for all distinct \(z, z' \in Z_p\).
  Then \(p \leq |X|-1\) and \(Z_p\) is the desired set.

  Let \(Z_0 \coloneqq X\).
  Now suppose that \(Z_i\) is defined.
  If \(\neighb{R_i}{G}{z} \cap \neighb{R_i}{G}{z'} = \emptyset\) for all distinct \(z, z' \in Z_i\),
  we set \(p \coloneqq i\) and stop the construction.
  Otherwise, let
  \(Z_{i+1} \subseteq Z_i\) be inclusion-wise maximal such that
  \(\neighb{R_i}{G}{z} \cap \neighb{R_i}{G}{z'} = \emptyset\) for all distinct \(z, z' \in Z_{i+1}\).
  Clearly \(Z_{i+1} \subset Z_{i}\), as otherwise we would have stopped the construction.
  Then for all \(y \in Z_i\) there is a \(z \in Z_{i+1}\) such that
  \(\neighb{R_i}{G}{y} \cap \neighb{R_i}{G}{z} \neq \emptyset\),
  which implies that \(\neighb{R_i}{G}{y} \subseteq \neighb{3R_i}{G}{z} = \neighb{R_{i+1}}{G}{z}\).
  Thus, \(\neighbr{G}{X} \subseteq \neighb{R_i}{G}{Z_i} \subseteq \neighb{R_{i+1}}{G}{Z_{i+1}}\).

  In the worst case, all pairs \(Z_i,Z_{i+1}\) differ by exactly one element and we actually need \(p=|X|-1\) which implies that \(|Z_p|=1\).
  This case can be reached when each \(x_i\) is at position \(3^{i-1}r\) on a path and \(x_1\in Z_i\) for all \(i\).
\end{proof}

\subparagraph*{Nowhere Dense Graphs}
Let \(G\) be a graph and \(r,s > 0\).
The \emph{\((r,s)\)-splitter game} on \(G\) is played by two players called
\emph{Connector} and \emph{Splitter}.
The game is played in a sequence of at most \(s\) \emph{rounds}.
We let \(G_0 \coloneqq G\).
In round \(i+1\) of the game, Connector picks a vertex \(v_{i+1} \in V(G_i)\),
and Splitter answers by picking \(w_{i+1} \in \neighbr{G_i}{v_i}\).
We let \(G_{i+1} \coloneqq G_i[\neighbr{G_i}{v_{i+1}} \setminus \{w_{i+1}\}]\).
Splitter \emph{wins} if \(G_{i+1} = \emptyset\); otherwise the play continues.

\begin{fact}[\cite{grokresie17}]
  A class \(\CC\) of graphs is nowhere dense if for every \(r > 0\) there is an \(s > 0\)
  such that for every \(G \in \CC\),
  Splitter has a winning strategy for the \((r,s)\)-splitter game on \(G\).
\end{fact}

In a modified version of the game, in each round Connector
not only picks a vertex \(v\), but also a new radius \(r' \leq r\),
and the game continues in \(N_{r'}(v)\).
Clearly, reducing the radius does not help Connector, and the fact remains true for the same \(s\).
Nevertheless, it will be convenient for us later to work with this modified version of the game.

A graph class is \emph{effectively nowhere dense} if it is nowhere dense
and $s$ is given by a computable function $f(r)$.
In this paper, we consider only effectively nowhere dense graph
classes.

\subparagraph*{Types}
Let us fix a vocabulary \(\tau\) (of coloured graphs) in the following and assume that all graphs are \(\tau\)-coloured.
\(\FO[\tau]\) denotes the set of all formulas of first-order logic of vocabulary \(\tau\),
and \(\FO[\tau,q]\) denotes the subset of \(\FO[\tau]\)
consisting of all formulas of quantifier rank up~to~\(q\).

The \emph{\(q\)-type} of a tuple \(\bar{v} = (v_1, \dots, v_k) \in \big(V(G)\big)^k\) of arity \(k\)
is the set \(\type{q}{G}{\bar{v}}\) of all \(\FO[\tau,q]\)-formulas \(\phi(\bar{x})\)
such that \(G \models \phi(\bar{v})\).
Moreover, for \(q,r \geq 0\), the \emph{local \((q,r)\)-type} of \(\bar{v}\) is the set
\(\ltype{q}{r}{G}{\bar{v}} \coloneqq \type{q}{\Neighbr{G}{\bar{v}}}{\bar{v}}\).
A \emph{\(k\)-variable \(q\)-type (of vocabulary \(\tau\))}
is a set \(\theta\) of \(\FO[\tau,q]\)-formulas whose free variables are among \(x_1, \dots, x_k\).
Since, up to logical equivalence, there are only finitely many \(\FO[\tau,q]\)-formulas whose free variables are among \(x_1, \dots, x_k\), we can view types as finite sets of formulas.
Formally, we syntactically define a normal form
such that for all \(\tau, q, k\),
there are only finitely many \(\FO[\tau,q]\)-formulas in normal form
with free variables among \(x_1, \dots, x_k\).
Furthermore, there is an algorithm that transforms every formula
into an equivalent formula in normal form
without increasing the quantifier rank.
Then, we view types as sets of formulas in normal form.
We denote the set of all \(k\)-variable \(q\)-types of vocabulary \(\tau\)
by \(\Types{\tau}{k}{q}\).

It is easy to see that for every
\(\FO[\tau,q]\)-formula \(\phi(x_1, \dots, x_k)\),
there is a set \(\Phi \subseteq \Types{\tau}{k}{q}\)
such that for all \(\tau\)-coloured graphs \(G\)
and all \(\bar{v} \in \big(V(G)\big)^k\),
\[
G \models \phi(\bar{v}) \iff \type{q}{G}{\bar{v}} \in \Phi.
\]

The following fact is a consequence of Gaifman's Theorem~\cite{Gai82}.

\begin{fact}\label{fact:local}
  For all \(q \geq 0\), there is an \(r \coloneqq r(q) \in 2^{\bigO(q)}\)
  such that for all \(k \geq 1\), all vocabularies \(\tau\),
  all \(\tau\)-coloured graphs \(G\), and all \(\bar{v},\bar{v}' \in \big(V(G)\big)^k\),
  if\/ \(\ltype{q}{r}{G}{\bar{v}} = \ltype{q}{r}{G}{\bar{v}'}\),
  then \(\type{q}{G}{\bar{v}} = \type{q}{G}{\bar{v}'}\).
\end{fact}

It is important for us to note that this \(r\) can be chosen to depend only on \(q\),
independent of the vocabulary \(\tau\).
From \cref{fact:local}, we directly get the following corollary.

\begin{corollary}\label{cor:local}
  Let \(\phi(x_1, \dots, x_k)\) be an \(\FO[\tau,q]\)-formula,
  and let \(r \coloneqq r(q)\) be chosen according to \cref{fact:local}.
  Then for every graph \(G\), there is a set \(\Phi\) of \(k\)-variable \(q\)-types
  such that for all \(\bar{v} \in \big(V(G)\big)^k\),
  \[
    G \models \phi(\bar{v}) \iff \ltype{q}{r}{G}{\bar{v}} \in \Phi.
  \]
\end{corollary}

\subparagraph*{Parameterized Complexity}
A \emph{parameterized problem} (over a finite alphabet \(\Sigma\))
is a pair \((L, \kappa)\), where \(L \subseteq \Sigma^*\)
and \(\kappa \colon \Sigma^* \to \NN\) is a polynomial-time computable function,
called \emph{parametrization} of \(\Sigma^*\).
If the parametrization \(\kappa\) is clear from the context, we can omit it.
We say that \((L, \kappa)\) is \emph{fixed-parameter tractable}
or \((L, \kappa) \in \FPT\) if there is an algorithm \(\Algorithm{L}\),
a computable function \(f \colon \NN \to \NN\),
and a polynomial \(p\) such that \(\Algorithm{L}\) decides \(L\)
and runs in time \(f\big(\kappa(x)\big) \cdot p(\abs{x})\) on input \(x \in \Sigma^*\).
The class \FPT{} essentially captures all tractable parameterized
problems.

An \emph{fpt Turing reduction}  from \((L, \kappa)\) to \((L',
\kappa')\) is an fpt algorithm with oracle access to \((L',
\kappa')\) solving \((L, \kappa)\) in such a way that on input $x$,
for all oracle queries $x'$, it holds that $\kappa'(x')\le
g(\kappa(x))$ for some computable function $g$. Clearly, if there is
an fpt Turing reduction from \((L, \kappa)\) to \((L',
\kappa')\) and \((L', \kappa') \in \FPT\) then \((L, \kappa) \in \FPT\).

For additional background, we refer the reader to \cite{dowfel13,flugro06}.

\section{Learning Problems}\label{sec:learningProblem}

Recall the supervised learning scenario we described in the
introduction. We are given a sequence
\( (\bar{v}_1, \lambda_1), \ldots, (\bar{v}_m, \lambda_m) \in V(G)^k
\times \{0,1\} \) of labelled examples over some graph $G$,
which we call the \emph{background graph}.  In the query learning
scenario, the graph represents a database instance.
We assume that the $(\bar{v}_i, \lambda_i)$ are drawn
independently from some unknown ``data generating'' probability
distribution ${\mathcal D}$ on $V(G)^k\times\{0,1\}$.  Note that we
make no assumptions about the distribution $\mathcal D$. In
particular, we do not assume that for every $\bar v$ there is exactly
one (or at most one) $\lambda\in\{0,1\}$ such that
$\mathcal D(\{(\bar v,\lambda)\})>0$. That is, there is not necessarily
an (unknown) target function that assigns a ``true'' value
$\lambda\in\{0,1\}$ to every $\bar v$. Of course there could be such a
function, then essentially $\mathcal D$ would be a distribution on
$V(G)^k$, but our setting is more general and admits uncertainty in
the data. This setting is sometimes referred to as \emph{agnostic learning}.
The special case where an unknown target function exists and is
actually one of the hypotheses is called the \emph{realisable} case;
we shall see in \cref{sec:hardness} that our hardness result even applies
in this special case.

Our hypotheses are first-order queries, formally represented as Boolean functions
$h_{\phi,\bar w}:V(G)^k\to\{0,1\}$ for some first-order formula $\phi(\bar x\smid\bar
y)$ and tuple $\bar w\in V(G)^\ell$.
Our goal is
to generate a hypothesis $h_{\phi,\bar w}$ that generalises well, that is, minimises the
generalisation error
\[
  \Pr_{(\bar v,\lambda)\sim{\mathcal D}}\big(h_{\phi,\bar w}(\bar v)=\lambda\big).
\]

It is well-known that, at least from a theoretical perspective, it
suffices to minimise the \emph{empirical risk} or \emph{training
  error}, that is, the fraction of training examples that are
classified wrong. This paradigm is known as \emph{empirical risk
  minimisation (ERM)}. The two main results relevant for us here are
the following. The \emph{Uniform Convergence Theorem} (see
\cite[Chapter 4]{Shalev-Shwartz:2014:UML:2621980}) assumes that we
have a finite hypothesis class ${\CH}$. It states that for all
$\epsilon,\delta>0$, it suffices to see $\onot(\log|{\CH}|)$
labelled examples to guarantee that with probability at least
$1-\delta$, the training error of a hypothesis $h\in{\CH}$
deviates from the generalisation error by at most $\epsilon$. This
implies that the generalisation error of a hypothesis $h\in{\CH}$
minimising the empirical risk deviates from the generalisation error of
the best possible hypothesis $h^*\in{\CH}$ by at most
$2\epsilon$. In other words: for finite hypothesis classes, ERM is
\emph{probably approximately correct}, or a \emph{PAC learning algorithm}. The second result, sometimes
called the \emph{Fundamental Theorem of PAC Learning} (see
\cite[Chapter 6]{Shalev-Shwartz:2014:UML:2621980}), lifts this to
infinite hypothesis classes ${\CH}$, as long as their \emph{VC
  dimension} $\operatorname{VC}({\CH})$ is finite. It states
that uniform convergence already holds if we see
$\onot\big(\!\operatorname{VC}(|{\CH}|)\big)$ examples, and hence ERM is
probably approximately correct in this more general setting. It will
not be necessary here to understand VC dimension. To tie things together, just note that for a
finite hypothesis class ${\CH}$, the VC dimension is at most
$\log|{\CH}|$.

The uniform convergence results relating the training error and the
generalisation error not only show that an ERM algorithm yields a PAC
learning algorithm, but conversely they also allow us to obtain a
(randomised) ERM algorithm from a PAC learning algorithm, simply by
applying the PAC learning algorithm with a data generating
distribution that is uniform on the training examples (and zero
everywhere else). Thus, ERM and PAC learning are essentially
equivalent, and therefore we focus on ERM from now on.

The hypothesis classes we are interested in here are classes
${\CH}_{k,\ell,q}(G)$, where $G$ is a graph and
$k,\ell,q\in\mathbb N$, consisting of all $h_{\phi,\bar w}$ for
first-order formulas $\phi(\bar x\smid \bar y)$ of quantifier rank $q$
with $|\bar x|=k$, $|\bar y|=\ell$, and tuples
$\bar w\in \big(V(G)\big)^\ell$. Since up to logical equivalence there are only
finitely many first-order formulas of quantifier rank $q$ with
$(k+\ell)$ free variables, these hypothesis classes are finite sets of size
$f(k,\ell,q)\cdot n^\ell$, where $n$ is the order of $G$. Thus, an
ERM-learning algorithm that sees $\onot(\log n)$ training examples is
a PAC learning algorithm. Furthermore, for graphs $G$ from a
nowhere dense class $\mathcal C$ of graphs, the VC dimension of the
hypothesis classes ${\CH}_{k,\ell,q}(G)$ is uniformly bounded
by a constant $d=d(\mathcal C,k,\ell,q)$, and therefore an
ERM-learning algorithm only needs to see $\onot(d)$ (independent of $n$)
examples to be a PAC learning algorithm.

To conclude: the algorithmic problem we need to solve is empirical
risk minimisation. We always denote the background graph by $G$ and its vocabulary by
$\tau$. We assume the input sequence of training examples is
\[
  \Lambda=\Big((\bar{v}_1, \lambda_1), \ldots, (\bar{v}_m, \lambda_m)\Big) \in\Big(V(G)^k
  \times \{0,1\}\Big)^m.
\]
The \emph{training error} of a hypothesis $h$ is
\[
  \err_\Lambda(h)\coloneqq\frac{1}{m}\big|\big\{i\in[m]\bigmid h(\bar
  v_i)\neq\lambda_i\big\}\big|.
\]
For hypotheses $h_{\phi,\bar w}$, we write $\err_\Lambda(\phi,\bar w)$
instead of $\err_\Lambda(h_{\phi,\bar w})$.
Here is the precise statement of the ERM problem.

\begin{problem}{\FOERM}
    \item[Input]
    Graph \(G\), training sequence $\Lambda\in \big(V(G)^k
  \times \{0,1\}\big)^m$, and $k,\ell,q\in\mathbb N$, $\epsilon>0$.
    \item[Parameter]
      \(k + \ell + q + \abs{\tau}+\frac{1}{\epsilon}\)
    \item[Problem]
    Return a hypothesis $h_{\phi,\bar w}\in\CH_{k,\ell,q}(G)$
    such that
    \[
      \err_{\Lambda}(\phi,\bar{w}) \le \epsilon^*+\epsilon,
    \]
    where
    $\epsilon^*\coloneqq\min\{\err_{\Lambda}(h)\mid h\in
    \CH_{k,\ell,q}(G)\}$.
\end{problem}

Note that we only require a hypothesis approximately minimising the
risk, because such an approximation is sufficient for PAC learning.

Unfortunately, in our positive learnability result for nowhere
dense graphs, we cannot
guarantee to match the quantifier rank and parameter number of the
optimal hypothesis, but may return a hypothesis with larger
quantifier rank and parameter number. To make this precise, we
introduce the following relaxation of the problem.
Let \(L,Q \colon \NN^3 \rightarrow \NN\) be functions with
\(L(k,\ell,q) \geq \ell\) and \(Q(k,\ell,q) \geq q\)
for all \(k,\ell,q \in \NN\).

\begin{problem}{\RFOERM{L,Q}}
    \item[Input]
    Graph \(G\), training sequence $\Lambda\in \big(V(G)^k
  \times \{0,1\}\big)^m$, and $k,\ell^*,q^*\in\mathbb N$, $\epsilon>0$.
    \item[Parameter]
      \(k + \ell^* + q^* + \abs{\tau}+\frac{1}{\epsilon}\)
    \item[Problem]
      Return a hypothesis $h_{\phi,\bar w}\in\CH_{k,\ell,q}(G)$ for some
    $q\le Q(k,\ell^*,q^*)$ and $\ell\le L(k,\ell^*,q^*)$
    such that
    \[
      \err_{\Lambda}(\phi,\bar{w}) \le \epsilon^*+\epsilon,
    \]
    where
    $\epsilon^*\coloneqq\min\{\err_{\Lambda}(h)\mid h\in
    \CH_{k,\ell^*,q^*}(G)\}$.
\end{problem}

 \section{The Hardness of Learning}
\label{sec:hardness}

In this section, we shall prove that the ERM problem is hard; at least
as hard as the following \emph{model-checking problem} for first-order
logic.

\begin{problem}{\FOmc}
  \item[Input]
    Graph \(G\), \FO-sentence \(\phi\)
  \item[Parameter]
    \(\abs{\phi}\)
  \item[Problem]
    Decide whether \(G \models \phi\) holds.
\end{problem}

As \FOmc{} is hard for the parameterized complexity class \AWstar{}
\cite{dowfeltay96}, this will prove \cref{theo:hardness}.

\begin{lemma}\label{lem:hardness}
  Let\/ \(L,Q \colon \NN^3 \to \NN\) with
  \(L(k,\ell,q) \geq \ell\) and \(Q(k,\ell,q) \geq q\)
  for all \(k, \ell, q\).
  Then \FOmc{} is fpt Turing reducible to \RFOERM{L,Q}.
\end{lemma}

\begin{proof}
  We show that there is an fpt algorithm with access to a \RFOERM{L,Q}
  oracle that solves \FOmc{}. To turn this
  into an fpt reduction, we need to be careful that the algorithm
  makes only oracle calls with parameter
  $k+\ell^*\!+q^*\!+\abs{\tau}+\frac{1}{\epsilon}\le g(p)$, where $p$ is the
  length of the input formula of the model-checking problem and $g$ is
  some computable function.

  We first give the proof under the assumption $L(1,0,q)=0$ for all
  $q$.

  Let \(G\) be a $\tau$-colored graph of order $n$, and let
  \(\phi\) be an \FO-sentence of length $p$. Let $q$ be the quantifier
  rank of $\phi$. Of course we have $q\le p$, and we may also assume
  that $|\tau|\le p$, because we can ignore relation symbols not
  occurring in $\phi$.

  Without loss of generality, we may assume
  \(\phi = \exists x\, \psi(x)\) for some \FO-formula~\(\psi(x)\) of
  quantifier rank at most \(q-1\).

  Using \(\binom{n}{2}\) calls to \RFOERM{L,Q}, we aim to partition the
  nodes $V(G)$ into a bounded number of classes in such a way that
  nodes in the same class have the same $(q\!-\!1)$-type.
  The number of classes of the partition will be bounded in terms of
  the parameter $p$.

  \medskip
  \begin{claim}\label{claim:hardness-distinguish}
    Let \(u,v \in V(G)\) with \(\type{q-1}{G}{u} \neq
    \type{q-1}{G}{v}\).
    Then \RFOERM{L,Q} on input \(\Lambda = \big((u,0),(v,1)\big)\)
    with
    parameters \(k=1\), \(\ell^*=0\), \(q^*=q-1\) and \(\epsilon=1/4\)
    returns a formula \(\phi'(x)\) of quantifier rank at most \(Q(1,0,q-1)\)
    such that \(G \not\models \phi'(u)\) and \(G \models \phi'(v)\).
  \end{claim}

  \begin{claimproof}
    Let \(\phi^*(x) = \bigwedge_{\gamma \in \type{q-1}{G}{v}} \gamma(x)\),
    which is an \FO-formula of quantifier rank \(q-1\).
    Then, \(G \not\models \phi^*(u)\) and \(G \models \phi^*(v)\).
    Thus, \(\err_{\Lambda}(\phi^*) = 0\).
    Since \(\epsilon = \frac{1}{4}\), \RFOERM{L,Q} has to return a formula
    \(\phi'(x)\) with \(\err_{\Lambda}(\phi') \leq \frac{1}{4}\). Note
    that here we use $L(k,0,q)=0$ to make sure that the formula
    $\phi'$ does not have additional parameters.
    Furthermore, by the definition of the training error,
    \(\err_{\Lambda}(\phi') \in \{0,\frac{1}{2},1\}\).
    Hence,
    \(\err_{\Lambda}(\phi') = 0\)
    and
    \(G \not\models \phi'(u)\), \(G \models \phi'(v)\).
  \end{claimproof}

  The difficulty of the proof is that we do not have a ``converse'' of
  Claim~1. If \(\type{q-1}{G}{u} = \type{q-1}{G}{v}\), then the
  \RFOERM{L,Q}-oracle still returns a formula \(\phi'(x)\), but we
  know nothing about this formula, and we have no way of recognising
  that we are in this case.

  As discussed in the preliminaries, by introducing a suitable normal form,
  we can always assume that the set $\FO[\tau, Q(1,0,q-1)]$ of
  $\FO$-formulas of vocabulary $\tau$ and quantifier rank at most
  $Q(1,0,q-1)$ is finite.
  Let \(s\) be an upper bound for the size of this
  set.
  Note that \(s\) can be bounded in terms of \(p\).

  Our next goal is to compute a ``bounded'' set $T\subseteq V(G)$ of
  representatives of the $(q-1)$-types, that is, we want to guarantee the
  following:
  \begin{romanenumerate}
    \item for every $v\in V(G)$ there is a $t\in T$ such that
      $\type{q-1}{G}{v} = \type{q-1}{G}{t}$;
    \item $|T|\le h(p)$ for some computable function $h$.
  \end{romanenumerate}
  To choose the function $h$ in (ii), we recall
  Ramsey's Theorem (see, for example, \cite{cam94}):
  {\itshape
    there is a computable function $R:\mathbb N^3\to\mathbb N$ such
    that for all $k,\ell,m$, every set
    $S=\{s_1,\ldots,s_{r}\}$ of size $r> R(k,\ell,m)$, and every mapping
    $C:\binom{S}{k}\to [\ell]$, there is a subset $M\subseteq S$ of
    size $|M|=m$ that is monochromatic, that is, the restriction of
    $C$ to $\binom{M}{k}$ is constant.
  }
  We let $h(p)\coloneqq R(2,s,3)$.

  We fix an arbitrary linear order $<$ on $V(G)$.
  For each
  unordered pair $\{u,v\}\in\binom{V(G)}{2}$, say with $u<v$, we
  compute an FO-formula $\gamma_{u,v} = \gamma_{u,v}(x)$
  of quantifier rank at most $Q(1,0,q\!-\!1)$ by calling \RFOERM{L,Q} on input
  $G$ and $\Lambda = \big((u,0),(v,1)\big)$ with parameters \(k=1\), \(\ell^*=0\),
  \(q^*=q-1\) and \(\epsilon=1/4\). By \cref{claim:hardness-distinguish},
  if \(\type{q-1}{G}{u}\neq\type{q-1}{G}{v}\), we have
  $G\not\models\gamma_{u,v}(u)$ and $G\models\gamma_{u,v}(v)$.

  \medskip
  \begin{claim}
    Let $S\subseteq V(G)$ of size $|S|> h(p)$. Then we can
    find three vertices $v_1,v_2,v_3 \in S$ such that either
    \(\type{q-1}{G}{v_1}=\type{q-1}{G}{v_2}\) or
    \(\type{q-1}{G}{v_2}=\type{q-1}{G}{v_3}\).
  \end{claim}

  \begin{claimproof}
    By Ramsey's Theorem, there is a monochromatic subset of size $3$,
    that is, three vertices $v_1,v_2,v_3$ such that
    $\gamma_{v_1,v_2}=\gamma_{v_1,v_3}=\gamma_{v_2,v_3}=:\gamma$. We
    shall prove that these vertices satisfy the assertion of the
    claim.

    Suppose $\type{q-1}{G}{v_1}\neq\type{q-1}{G}{v_2}$; if they are
    equal, there is nothing to prove. Then $G\not\models\gamma(v_1)$
    and $G\models\gamma(v_2)$ by Claim~1.
    If in addition we had $\type{q-1}{G}{v_2}\neq\type{q-1}{G}{v_3}$,
    this would imply $G\not\models\gamma(v_2)$
    (and $G\models\gamma(v_3)$), which is a contradiction. Thus,
    $\type{q-1}{G}{v_2}=\type{q-1}{G}{v_3}$.\footnote{Note that we do not need the full strength of Ramsey's
     Theorem here, because we never use $\gamma_{v_1,v_3}=\gamma$. It
     is not hard to prove the weaker statement that we need directly,
     but we found it most convenient to just apply Ramsey's Theorem.}
 \end{claimproof}

 We can use the claim to iteratively construct a set $T$ satisfying
 (i) and (ii). We construct a sequence $T_0\supset T_1\supset\cdots$
 of subsets of $V(G)$ that satisfy (i) and stop as soon as we arrive
 at a set $T:=T_i$ satisfying (ii). Let $T_0\coloneqq V(G)$. If
 $|T_i|>h(p)$, we find $v_1,v_2,v_3$ according to Claim~2 and let
 $T_{i+1}\coloneqq T_i\setminus\{v_2\}$.

 Recall that we assumed the input formula $\phi$ of the model-checking
 problem to be of the shape $\exists x\,\psi(x)$. Since $\psi(x)$ is a
 formula of quantifier rank at most $q-1$, for all $u,v$ with
 $\type{q-1}{G}{u}=\type{q-1}{G}{v}$ it holds that
 $G\models\psi(u)\iff G\models\psi(v)$. Thus by (i), $G\models\phi$ if
 and only if there is a $t\in T$ such that $G\models \psi(t)$. This
 means that to find out if $G\models\phi$, we
 only need to check if $G\models\psi(t)$ for $t\in
 T$. We can do this by recursively calling our algorithm $|T|$ times.

 The only small difficulty is that the input formula to the
 model-checking problem needs to be a sentence, that is, a formula
 without free variables. But we can easily resolve this by
 constructing, for each $t\in T$, a sentence $\psi_t$ and an expansion
 $G_t$ of $G$ such that
 $G\models\psi(t)\iff G_t\models\psi_u$ and then making the recursive
 call to $G_t,\psi_t$. We do this by adding two fresh unary relation
 symbols $P_t,Q_t$, which are interpreted by the sets $P_t(G):=\{t\}$
 and $Q_t(G):=\{u\mid u\in N(t)\}$ in $G_t$. To obtain $\psi_t$
 from $\psi(x)$, we replace all atoms $x=y$, $y=x$ by $P_t(y)$ and all atoms
 $E(x,y)$, $E(y,x)$ by $Q_t(y)$, assuming without loss of generality that there
 are no atoms $E(x,x)$ and $x=x$.

  The recursion tree has depth at most \(q\)
  and the degree is bounded by a function of $p$. Thus, the size of the
  recursion tree is bounded in terms of $p$. Each recursive call
  requires only fpt time, so overall we have an fpt algorithm.

  This completes the proof of the lemma under the assumption that
  $L(1,0,q)=0$ for all $q$.

  The only place where we used the assumption $L(1,0,q)=0$ is in
  the proof of Claim~1. So let us focus on this claim.  Again, let
  \(u,v \in V(G)\) such that \(\type{q-1}{G}{u} \neq
  \type{q-1}{G}{v}\). It is our goal to find a formula $\phi'(x)$
  (without additional free variables) such that $G\not\models\phi'(u)$
  and $G\models\phi'(v)$. The quantifier rank of this formula needs to
  be bounded in terms of $p$, but there is no need for it to be at most
  \(Q(1,0,q-1)\).

  Instead of using types, here we use local types to distinguish $u$
  and $v$. By \cref{fact:local}, we have
  $\ltype{q-1}{r}{G}{u}\neq\ltype{q-1}{r}{G}{v}$ for some $r=r(q-1)\in
  2^{\bigO(q)}$. We can use this to construct an FO-formula
  $\phi^*(x)$ that such that $G\not\models\phi^*(u)$,
  $G\models\phi^*(v)$ and the formula $\phi^*$ is \emph{$r$-local}, which
  means that for all graphs $H$ and all vertices $w\in V(H)$ we have
  $H\models\phi^*(w)$ if and only if
  $\Neighbr{H}{w}\models\phi^*(w)$. To achieve this, we simply need to
  restrict all quantifiers to vertices of distance at most $r$ from
  $x$. This implies that the quantifier rank $q^*$ of the formula
  $\phi^*(x)$ will be in $\bigO(\max\{q,\log r\})=\bigO(q)$.

  Let $\ell:=\max\{1,L(k,0,q^*)\}$, and let $\widehat G$ be the union of
  $2\ell$ disjoint copies $G^{(1)},\ldots,G^{(2\ell)}$ of $G$. For
  every vertex $w\in V(G)$, let $w^{(i)}$ be the copy of $w$ in
  $G^{(i)}$. Let
  $\widehat\Lambda\coloneqq\big((u^{(i)},0),(v^{(i)},1)\bigmid
  i\in[2\ell]\big)$. Observe that for all $i\in[2\ell]$, it holds that
  $\widehat G\not\models\phi^*(u^{(i)})$ and
  $\widehat G\models\phi^*(v^{(i)})$, because
  $\Neighbr{\widehat G}{u^{(i)}}=\Neighbr{G}{u}$ and
  $\Neighbr{\widehat G}{v^{(i)}}=\Neighbr{G}{v}$. Thus
  $\err_{\widehat\Lambda}(\phi^*)=0$.

  We apply \RFOERM{L,Q} on input $\widehat G,\widehat\Lambda$ with parameters
  \(k=1\), \(\ell^*=0\), \(q^*\) and \(\epsilon=1/8\) and obtain a
  formula $\phi(x\smid \bar y)$ with $\ell'=|\bar y|\le\ell$ and a
  tuple $\bar w=(w_1,\ldots,w_{\ell'})\in V(\widehat G)^{\ell'}$ such that
  $\err_{\widehat\Lambda}(\phi,\bar w)\le \frac{1}{8}$. Let us call an
  index $i\in[2\ell]$ \emph{covered} if there is parameter
  $w_i\in V(G^{(i)})$. Moreover, we call $i$ \emph{wrong} if either
  $\widehat G\models\phi(u^{(i)}\smid\bar w)$ or
  $\widehat G\not\models\phi(v^{(i)}\smid\bar w)$. Note that at most
  $\ell'\le\ell$ indices are covered and at most $|\Lambda|/8=\ell/2$
  are wrong. Thus, as $\ell\ge 1$, there is an index $i^\circ$ that is
  neither covered nor wrong. In the following, we let
  $u^\circ\coloneqq u^{(i^\circ)}$ and $v^\circ\coloneqq v^{(i^\circ)}$

  We expand $\widehat G$ by unary relations $P_i$ to a graph $\widehat G'$
  with $P_i(\widehat G')=\{w_i\}$. We can easily construct a formula
  $\phi'(x)$ such that for all $\widehat v\in V(\widehat G)=V(\widehat G')$ we
  have $\widehat G\models\phi(\widehat v\smid\bar w)\iff\widehat
  G'\models\phi'(\widehat v)$. In particular, we have $\widehat
  G'\not\models\phi'(u^{\circ})$ and $\widehat
  G'\models\phi'(v^\circ)$.

  From $\phi'(x)$ we can construct an
  $r'$-local formula $\phi''(x)$, for some $r'$ bounded in terms of the quantifier rank of
  $\phi''$, such that for all $\widehat v\in V(\widehat G')$ we
  have $\widehat
  G'\models\phi'(\widehat v)\iff \widehat
  G'\models\phi''(\widehat v)$.
  Again, $\widehat
  G'\not\models\phi''(u^{\circ})$ and $\widehat
  G'\models\phi''(v^\circ)$.

  Now let $\phi'''(x)$ be the formula (of the original vocabulary
  $\tau$) obtained from $\phi''(x)$ by replacing each atomic
  subformula $P_i(z)$ by \textsc{False} (or $\neg z=z$). Then $\widehat
  G\not\models\phi'''(u^{\circ})$ and $\widehat
  G\models\phi'''(v^\circ)$, simply because there is no $w_i$ in the
  $r'$-neighbourhood of $u^\circ$ or $v^\circ$. Since $\phi'''(x)$ is
  $r'$-local and $\mathcal N_{r'}^{\widehat G}(u^{\circ})=\mathcal N_{r'}^{G}(u)$ and
  $\mathcal N_{r'}^{\widehat G}(v^{\circ})=\mathcal N_{r'}^{G}(v)$, it follows that
  $G\not\models\phi'''(u)$ and  $G\models\phi'''(v)$. Thus
  $\phi'''(x)$ has the desired property.

  Using this generalisation of Claim~1, we can complete the proof as
  in the case ${L(1,0,q)=0}$.
\end{proof}

\begin{remark}
  The proof only uses that the \RFOERM{L,Q}-oracle delivers correct
    answers in the case that $\epsilon^*=0$, that is, there is a
    hypothesis consistent with the training data. This means that the
    hardness result even holds for the realisable case of the learning
    problem.
\end{remark}

 \section{Fixed-parameter Tractable Learning}
\label{sec:fpt}

For a class \(\CC\) of graphs and a parameterized problem \(M\),
with \(M\) being \FOERM, \RFOERM{L,Q}, or \FOmc,
we say that \(M\) is \emph{fixed-parameter tractable on \(\CC\)}
if the restriction of \(M\) to input graphs from \(\CC\)
is fixed-parameter tractable.

Let us say that a class $\CC$ of (coloured) graphs is closed under
\emph{colour expansions} if for all $\tau$-coloured graphs $G$, all
$\tau'\subseteq\tau$, and all $\tau'$-expansions $G'$ of $G$, if
$G\in\CC$ then $G'\in\CC$.  It can easily be checked that the
reduction from \FOmc{} to \RFOERM{Q,L} that we gave in the previous
section can be restricted to all graph classes that are closed under
colour expansions and disjoint unions. As \RFOERM{Q,L} trivially
reduces to \FOERM{}, the result also applies to \FOERM{}.

In the following two lemmas, we give
converse reductions from \FOERM{} to \FOmc{}, but only under additional
restrictions on the parameters $k$ and $\ell$. Again, since
\RFOERM{Q,L} reduces to
\FOERM{}, the results also apply to \RFOERM{Q,L}.

\begin{restatable}{proposition}{constantL}
\label{lem:constant-l}
  Let \,\(\CC\) be a class of graphs closed under colour expansions such that \FOmc{} is fixed-parameter
  tractable on \(\CC\).

  Then, for every constant \(\ell^* \in \NN\),
  the restriction of \FOERM\ to inputs with $\ell=\ell^*$
  is fixed-parameter tractable on \(\CC\).
\end{restatable}

For the proof of \cref{lem:constant-l}, we use that a simple brute-force test
of all \(|V(G)|^{\ell^*}\) possible parameter configurations
can be done within the given time bound.
A formal proof is given in the appendix.

Recall that we call a learning problem \emph{realisable} if there
is a target function, and this target function is contained in the
hypothesis class. In the realisable version of the ERM problem, we
assume that the training error $\epsilon^*$ of the best possible
hypothesis is~$0$. Note that formally, the restriction of \FOERM{} to
the realisable case is a \emph{promise problem} where an
algorithm is only required to give a correct answer if the input
satisfies a certain assumption, the ``promise'', in our case $\epsilon^*=0$. On inputs that do
not fulfil the promise, the algorithm can do anything. The following
proposition considers the realisable version of \FOERM{} restricted to
the 1-dimensional case $k=1$.

\begin{restatable}{proposition}{kIsOne}
\label{lem:k-one}
  Let \,\(\CC\) be a class of graphs closed under colour expansions such that \FOmc{} is fixed-parameter
  tractable on \(\CC\).\\
Then
  the following restriction of \FOERM{}
  is fixed-parameter tractable.

\begin{center}
  \begin{problem}{}
    \item[Input]
    Graph \(G\in\CC\), training sequence $\Lambda\in \big(V(G)^k
  \times \{0,1\}\big)^m$, and $\ell,q\in\mathbb N$, $\epsilon>0$.
    \item[Parameter]
      \(\ell + q + \abs{\tau}+\frac{1}{\epsilon}\)
    \item[Assumption] There is a  $h\in\CH_{1,\ell,q}(G)$ with $\err_{\Lambda}(h)=0$.
    \item[Problem]
    Return a hypothesis $h_{\phi,\bar w}\in\CH_{1,\ell,q}(G)$
    such that
    \[
      \err_{\Lambda}(\phi,\bar{w}) \le \epsilon.
    \]
\end{problem}
\end{center}
\end{restatable}

Here, techniques similar to those used in
\cite{GroheLoedingRitzert_MSO,GrienenbergerRitzert_Trees}
can be applied to find a consistent parameter setting.
Again, a formal proof is given in the appendix.

The following result is a precise version of \cref{theo:tractability}.
\begin{theorem}\label{thm:fpt-learning}
  For every nowhere dense class $\CC$ of graphs,
  there are functions \(L, Q \colon \NN^3 \to \NN\)
  such that the problem \RFOERM{L,Q}
  is fixed-parameter tractable on
  \(\mathcal C\).
\end{theorem}

The remainder of this section is devoted to the proof of the theorem.
In the following, we let \(\CC\) be an effectively nowhere dense graph class,
and we fix \(k \geq 1\) and \(\ell^*, q^* \geq 0\).
We choose \(r = r(q^*)\) according to \cref{fact:local} and let
\(R \coloneqq 3^{\ell^*-1} \cdot \big((k+2)(2r+1)\big)\).
The specific choice of \(R\) will be justified in the proof of
\cref{lem:projection} in the appendix.
Let \(G \in \CC\) and let \(s \geq 1\) such that Splitter has
a winning strategy for the \((R,s)\)-splitter game on \(G\).
The value \(s\) can be computed since \(\CC\) is effectively nowhere dense.
We let \(V \coloneqq V(G)\), \(E \coloneqq E(G)\), \(n \coloneqq \abs{V}\),
\(\ell \coloneqq L(k,\ell^*,q^*) \coloneqq \ell^* \cdot s\) and
\(q \coloneqq Q(k,\ell^*,q^*) \coloneqq q^* + \log R\).
Furthermore, let \(\Lambda\in \big(V(G)^k \times \{0,1\}\big)^m\)
be a training sequence and let
\(\Lambda^+ = \set{\bar v \in V^k \mid (v,1) \in \Lambda}\) and
\(\Lambda^- = \set{\bar v \in V^k \mid (v,0) \in \Lambda}\)
be the sets of positive and negative examples.
Let \(\epsilon^* \geq 0\) such that there is an \FO-formula
\(\phi^*(\bar{x}\smid \bar{y})\) of quantifier rank at most \(q^*\)
with \(k + \ell^*\) free variables and a tuple
\(\bar{w}^* = (w_1^*, \dots, w_{\ell^*}^*) \in V^{\ell^*}\) such that
\(\err_{\Lambda}(\phi^*,\bar{w}^*) \leq \epsilon^*\).

Let \(\epsilon > 0\).
Whenever we speak of an \emph{fpt algorithm} in the following,
we mean an algorithm running in time
\(f(k, \ell^*, q^*, s, 1/\epsilon) \cdot (n+m)^{\bigO(1)}\).
Our goal is to compute a query \((\phi, \bar{w})\), our \emph{hypothesis},
consisting of an \FO-formula \(\phi(\bar{x}\smid \bar{y})\)
of quantifier rank at most \(q\) with \(k + \ell\) free variables
and a tuple \(\bar{w} \in V^\ell\) such that
\(\err_{\Lambda}(\phi,\bar{w}) \leq \epsilon^* + \epsilon\),
using an fpt algorithm.
This means that the hypothesis correctly classifies all but
\(m \cdot (\epsilon^*+\epsilon)\) examples from \(\Lambda\).

A vector \(\bar{v} \in V^{\ell'}\) is called \emph{\(\epsilon'\)-discriminating}
if there is a formula \(\psi\) of quantifier rank at most \(q\) such that
\(\err_{\Lambda}(\psi,\bar{v}) \leq \epsilon'\).
We then say that \(\bar v\) \emph{discriminates} all but (at most)
\(\epsilon' \cdot m\) examples.
If we have an \((\epsilon^* + \epsilon)\)-discriminating
\(\ell\)-tuple \(\bar{w}\),
then we can find the formula \(\phi\) by an fpt algorithm
that simply steps through all possible formulas
since model checking on nowhere dense graphs is in \FPT{}
and there are only finitely many formulas to check.
Hence, in the following, we describe a procedure to find such an
\((\epsilon^* + \epsilon)\)-discriminating \(\ell\)-tuple.

A \emph{conflict} in \(\Lambda\) is a pair
\((\bar{v}^+,\bar{v}^-) \in \Lambda^+ \times \Lambda^-\) with
\(\ltype{q^*}{r}{G}{\bar{v}^+} = \ltype{q^*}{r}{G}{\bar{v}^-}\).
The \emph{type} of a conflict \((\bar{v}^+,\bar{v}^-)\) is the local
\((q^*,r)\)-type \(\ltype{q^*}{r}{G}{\bar{v}^+} = \ltype{q^*}{r}{G}{\bar{v}^-}\).
Let \(\Xi\) be the set of all conflicts.
To \emph{resolve} a conflict \((\bar{v}^+,\bar{v}^-) \in \Xi\),
we need to find parameters \(\bar{w}\) such that
\(\ltype{q^*}{r}{G}{\bar{v}^+ \bar{w}} \neq \ltype{q^*}{r}{G}{\bar{v}^- \bar{w}}\).
Only parameters in the \((2r+1)\)-neighbourhood of
\(\bar{v}^+ \cup \bar{v}^-\) have an effect on the local type.
We say that we \emph{attend} to the conflict \((\bar{v}^+, \bar{v}^-)\)
if we choose at least one parameter in \(\neighb{2r+1}{G}{\bar{v}^+ \bar{v}^-}\).
Note that attending to a conflict is a necessary, but not a sufficient condition
for resolving the conflict.
If we do not attend to a conflict, we \emph{ignore} it.
An example \(\bar{v} \in \Lambda\) is \emph{critical}
if it is involved in some conflict, that is,
\(\bar{v} \in \Lambda^+\) and there is a \(\bar{v}^- \in \Lambda^-\)
such that \((\bar{v},\bar{v}^-) \in \Xi\),
or \(\bar{v} \in \Lambda^-\) and there is a \(\bar{v}^+ \in \Lambda^+\)
such that \((\bar{v}^+,\bar{v}) \in \Xi\).
Let \(\Gamma\) be the set of all critical tuples.
For every \(w \in V\), let
\[
  \Gamma(w) \coloneqq \big\{ \bar{v} \in \Gamma \mid w \in \neighb{2r+1}{G}{\bar{v}} \big\}
\]
Note that if \(\bar{v} \in \Gamma(w)\), then \(w\) attends to all conflicts
\(\bar{v}\) is involved in.

The algorithm \(\Algorithm{L}\) now uses \(s\) steps,
corresponding to moves in the splitter game,
to find  a tuple that discriminates all but \(m (\epsilon^*+\epsilon)\) examples.
For this, we define a sequence of graphs \(G^0,\dots,G^{s}\)
and training sequences \(\Lambda^0, \dots, \Lambda^s\) with
\(\Lambda^i \in \big(V(G^i)^k \times \{0,1\}\big)^{m_i}\).
Let \(G^0 = G\) and \(\Lambda^0 = \Lambda\).
We will define the graphs \(G^i\) for \(i>0\)
and their training sequences \(\Lambda^i\) later.
Every tuple that is not critical can be classified correctly
without additional parameters by adding its local type explicitly to the hypothesis.
Thus, in every step, we consider only examples that have been involved in a conflict
in all previous steps.

We now consider the \(i\)-th step of the algorithm.
Let \(G^i\) be the current graph and
\(\Gamma^i\) be the set of critical examples in the \(i\)-th step.
In the following two lemmas,
we strategically limit the search space for the parameters
such that in each step, the error increases only by \(\frac{\epsilon}{s}\).
The next lemma shows that there is a small set~\(X\) of neighbourhood centres
such that most of the conflicts can be attended by nodes from \(\neighb{4r+2}{G^i}{X}\).
\begin{lemma}\label{lem:setX}
  There is a set \(X \subseteq V(G^i)\)
  of size \(\abs{X} \leq \frac{k \ell^* s}{\epsilon}\)
  such that
  \[
    \frac{\abs{\Gamma^i(u)}}{\abs{\Gamma}} < \frac{\epsilon}{\ell^*\cdot s}
    \qquad\text{ for all }u \in V(G^i) \setminus \neighb{4r+2}{G^i}{X},
  \]
  where \(\Gamma^i(u)\) is the set of critical tuples in \(G^i\)
  that are affected by \(u\)
  and \(\Gamma\) is the set of critical tuples in the original graph \(G\).
  Furthermore, there is an fpt algorithm computing such a set.
\end{lemma}

\begin{proof}
  We inductively define a sequence \(x_1, \dots, x_p \in V(G^i)\) as follows.
  For all \(j \geq 1\), we choose \(x_j \in V(G^i)\)
  such that \(\dist(x_j, x_{j'}) > 4r+2\) for all \(j' < j\) and,
  subject to this condition, \(\abs{\Gamma^i(x_j)}\) is maximum.
  If no such \(x_j\) exists, we let \(p = j-1\) and stop the construction.
We note that for every critical tuple \(\bar{v} \in \Gamma\),
  there are at most \(k\) of the \(x_j\) such that \(\bar{v} \in \Gamma^i(x_j)\).
  This holds as for every entry \(v\) of \(\bar{v} \in \big(V(G^i)\big)^k\),
  there is at most one of the \(x_j\) in the
  \((2r+1)\)-neighbourhood of \(v\) by the construction of \(X\).
  Therefore,
  \[
    \sum_{j=1}^p \abs{\Gamma^i(x_j)} \leq k \abs{\Gamma}.
  \]
  This means that there are at most \(k\ell^* s/\epsilon\) of the \(x_j\)
  with \(\abs{\Gamma^i(x_j)} \geq \frac{\epsilon}{\ell^* s}\abs{\Gamma}\).
  As the \(x_j\) are sorted by decreasing \(\abs{\Gamma^i(x_j)}\),
  we have \(\abs{\Gamma^i(x_j)} < \frac{\epsilon}{\ell^* s}\abs{\Gamma}\)
  for all \(j > \frac{k \ell^* s}{\epsilon}\).
  We let \mbox{\(X=\{x_1, \dots, x_{\min\{p, k \ell^* s/\epsilon\}}\}\)}.
\end{proof}

In the following, we fix a set \(X\) according to the lemma.
We show in the next lemma that there is always a tuple of parameters
consisting only of vertices in the neighbourhood of~\(X\)
that can induce a small error.
\begin{restatable}{lemma}{existenceOfDiscriminating}
\label{lem:existenceOfDiscriminating}
  If there is a tuple \(\bar v \in \big(V(G^i)\big)^{\ell^*}\)
  that discriminates all but \(m\cdot  \epsilon'\) examples in~\(G^i\),
  then there is a tuple \(\bar w \in \big(\neighb{4r+2}{G^i}{X}\big)^{\ell^*}\)
  that discriminates all but \(m(\epsilon'+\epsilon/s)\) examples in \(G^i\).
\end{restatable}
In the proof, \(\bar w\) is computed from \(\bar v\)
by dropping all elements that are far from \(X\).
By the choice of \(X\), this introduces only a small error.

\begin{proof}
  We split the entries \(v_j\) of \(\bar v\) in two sets \(W^*, U^*\),
  where \(W^*\) contains all \(v_j\) that are contained in
  \(\neighb{4r+2}{G^i}{X}\)
  and \(U^*\) contains the remaining entries of \(\bar v\).
  By the choice of \(X\), each \(u \in U^*\) can only discriminate up to
  \[
  \abs{\Gamma^i(u)}
    \leq \frac{m\abs{\Gamma^i(u)}}{\abs{\Gamma}}
    \leq \frac{m\epsilon}{\ell^* \cdot s}
  \]
  examples, where \(\Gamma^i(u)\) is the set of critical tuples in \(G^i\)
  that are affected by \(u\) and \(\Gamma\)
  is the set of critical examples in the original graph \(G\).
  Thus, the elements of $U^*$ can discriminate at most
  \(m \frac{\epsilon}{s}\) examples in \(G^i\).
  Since \(\bar{v}\) discriminates all but \(m\cdot \epsilon'\) examples,
  the elements of \(W^*\) must discriminate all but
  \(m \cdot (\epsilon' + \frac{\epsilon}{s})\) examples.
  We let \(\bar{w}\) be an \(\ell^*\)-tuple formed by the elements of \(W^*\),
  for example choosing \(w_j = v_j\) for all \(v_j \in W^*\),
  and \(w_j\) is set to a fixed element in \(W^*\) for all other entries.
  Then, \(\bar w \in \big(\neighb{4r+2}{G^i}{X}\big)^{\ell^*}\)
  and \(\bar{w}\) discriminates all but \(m(\epsilon'+\epsilon/s)\)
  examples in \(G^i\).
\end{proof}

We can now describe Step \(i\) and the way it is embedded in the overall algorithm.
We know that in \(G^0\), the vector \(\bar w^*\) discriminates all but
\(m\epsilon^*\) examples by the definition of \(\bar w^*\) and \(\epsilon^*\).
In all later graphs \(G^i\),
there is a vector that discriminates all but
\(m\cdot(\epsilon^*+\frac{i\epsilon}{s})\) of the examples \(\Lambda^i\)
which will be guaranteed by the construction of \(G^i\) and \(\Lambda^i\)
in \cref{lem:projection}.
By \cref{lem:setX},
there is also a tuple close to \(X\) that discriminates all but
\(m\cdot\big(\epsilon^*+\frac{(i+1)\epsilon}{s}\big)\) examples.
Thus, by further restricting the set of possible parameters,
we increase the error in each step.
In Step \(i\) we will choose \(\ell^*\) parameters.
Together with the parameters we choose in all following steps,
this will result in the wanted tuple that discriminates all but (at most)
\(m(\epsilon^* + \epsilon)\) examples.
Our next goal is to find such a tuple
\(\bar{w} \in \big(\neighb{4r+2}{G^i}{X}\big)^{\ell^*}\) that
(together with the parameters from the following steps) discriminates all but
\(m(\epsilon^* + \epsilon)\) examples in \(G^i\).

Clearly, for each such tuple of arity \(\ell^*\),
there is a subset \(Y \subseteq X\) of size \(\abs{Y} \leq \ell^*\)
such that \(\bar{w} \in \big(\neighb{4r+2}{G^i}{Y}\big)^{\ell^*}\).
We non-deterministically guess such a set \(Y = \{y_1, \dots, y_{\ell'}\}\)
and keep it fixed in the following.
Simulating this non-deterministic guess by a deterministic algorithm
adds a multiplicative cost of
\(\abs{X}^{\ell^*} \leq \big(\frac{k s \ell^* }{\epsilon}\big)^{\ell^*}\)\!,
which is allowed in an fpt algorithm.

When searching for parameters in \(\neighb{4r+2}{G^i}{Y}\),
we can attend only to conflicts with at least one element in \(\neighb{6r+3}{G^i}{Y}\) as all other conflicts cannot be attended by parameters from \(\neighb{4r+2}{G^i}{Y}\).
We apply \cref{lem:vitali} and obtain a set \(Z\subseteq Y\) and an
\(R' = 3^j \big((k+2)(2r+1)\big)\),
where \(0 \leq j \leq \abs{Y}-1 \leq \ell^*-1\),
such that
\[
\neighb{R'}{G^i}{z} \cap \neighb{R'}{G^i}{z'} = \emptyset
\]
for all distinct \(z,z' \in Z\) and
\[
\neighb{(k+2)(2r+1)}{G^i}{Y} \subseteq \neighb{R'}{G^i}{Z}.
\]
Note that \(R' \leq R\),
where \(R\) is the radius from the \((R,s)\)-splitter game
defined in the very beginning of the proof.
We will see in the proof of \cref{lem:projection} why we need
\(\big((k+2)(2r+1)\big)\)-neighbourhoods.
Suppose that \(Z=\{z_1, \dots, z_{\ell''}\}\),
where \(\ell'' \leq \ell' \leq \ell^*\).
For every \(j \in [\ell'']\),
let~\(w_j\) be Splitter's answer if Connector picks \(z_j\) together with
the radius \(R'\) in the (modified) \((R,s)\)-splitter game on \(G\).
Note that we consider only possible picks \(z_j\in Z\) and not arbitrary choices.
The reason why Splitter's answers to the \(z_j\) suffice is that
if we can identify every node in \(\neighb{R'}{G}{Z}\), then we can solve
(almost) all conflicts consisting of vectors of vertices from that set.
Essentially, Splitter guarantees that after removing a certain set of points
(her answers in the splitter game), every node in the neighbourhood
can be identified in at most \(s-(i+1)\) steps if she had a winning strategy
for the current graph in \(s-i\) steps (which is guaranteed by
the construction of \(G^i\)).

We then choose the vertices \(\hat w^i = (w_1, \dots, w_{\ell^*})\)
as parameters in step \(i\),
where \(w_j \coloneqq w_{\ell''}\) for all \(j \in \{\ell''+1, \dots, \ell^*\}\).
With those parameters, we can now define the next graph \(G^{i+1}\)
and the next set of examples \(\Lambda^{i+1}\).

\begin{restatable}{lemma}{projection}
\label{lem:projection}
  There is a graph \(G^{i+1}\) and a training sequence
  \(\Lambda^{i+1} \in \big(V(G^{i+1})^k \times \{0,1\}\big)^{m_{i+1}}\)
  for some \(m_{i+1} \leq m_i\) such that the following holds.
  \begin{bracketenumerate}
    \item \(V(G^{i+1}) = \neighb{R'}{G^i}{Z} \uplus U \uplus U^*\),
      where \(U\) is a set of fresh isolated vertices in \(G^{i+1}\) and
      \(U^*\) is the set of all isolated vertices in \(G^i\).
      \(\abs{U}\) depends only on \(k,\ell^*,q^*\)
      but not on \(m\) or \(n\). \label{lem:proj:nodes}
    \item \(E(G^{i+1}) \subseteq E(G^i)\).\label{lem:proj:edges}
    \item Splitter has a winning strategy for the
      \(\big((R,s-(i+1)\big)\)-splitter game on \(G^{i+1}\).
      \label{lem:proj:winning}
    \item If there is a tuple \(\bar w^i\) that discriminates
      all but \(m\cdot \epsilon'\) examples in \(G^i\), then there is a tuple
      \(\bar w^{i+1}\) that discriminates all but \(m(\epsilon'+\frac{\epsilon}{s})\) examples from \(\Lambda^{i+1}\)
      in \(G^{i+1}\).
      \label{lem:proj:down}
    \item If there is a tuple \(\bar{w}^{i+1}\) from \(V(G^{i+1})\)
      that distinguishes all but \(c\) examples from
      \(\Lambda^{i+1}\), then the tuple \(\hat{w}^i \bar{u}^{i+1}\)
      distinguishes all but \(c\) examples from \(\Lambda^{i}\) where
      \(\bar {u}^{i+1}\) is derived from \(\bar w^{i+1}\) by dropping all entries not in \(V(G^i)\).\label{lem:proj:up}
  \end{bracketenumerate}
  Furthermore, there is an fpt algorithm computing \(G^{i+1}\) and
     \(\Lambda^{i+1}\).
\end{restatable}

The formal proof of the lemma is given in the appendix.
From the lemma itself and its proof, we observe the following.

\begin{remark}
   For the graph \(G^{i+1}\) the following holds:
\begin{enumerate}
  \item The isolated vertices in \(G^{i+1}\) will not be suitable parameters,
    as any conflict they resolve could also be resolved without the use of additional parameters.
    Thus, we can assume that all nodes from \(\bar w^i\)
    also appear in \(G^i\) and we have that \(\bar u^i = \bar w^i\)
    in Statement \eqref{lem:proj:up}.
  \item Every conflict in the examples \(\Lambda^{i+1}\) in the graph \(G^{i+1}\)
    corresponds to a conflict between examples from \(\Lambda^i\) in the graph \(G^i\).
    The number of conflicts can thus only decrease.
  \item \(G^{i+1}\) is nowhere dense because Splitter has a winning strategy.
  \item \(\Lambda^{i+1}\) contains exactly the critical examples from \(\Lambda^i\).
  \item Examples in \(\Lambda^{i+1}\) are projections of examples in
    \(\Lambda^i\) into \(V(G^{i+1})\),
    essentially changing all those nodes that are not included in
    \(\neighb{R'}{G^i}{Z}\) to isolated vertices that represent types.
    The type each isolated vertex represents is encoded by its colour.
\end{enumerate}
\end{remark}

Note that only the last two items in the remark do not directly follow from
the lemma (and are details of its proof).
With the construction of the graph \(G^{i+1}\) and the corresponding set of
examples \(\Lambda^{i+1}\), we now have all the ingredients to prove
\cref{thm:fpt-learning}.

\begin{proof}[Proof of \cref{thm:fpt-learning}]
  We start the algorithm by computing the number of steps \(s\)
  in the splitter game.
  We set \(G^0=G\) and \(\Lambda^0= \Lambda\), as described above.
  For \(i=0\) to \(s\), we perform the following steps.

  We use \cref{lem:setX} to compute a set \(X\).
  By \cref{lem:existenceOfDiscriminating}, we know that there exists
  a tuple in the neighbourhood of \(X\) that discriminates all but
  \(m\big(\epsilon^*+\frac{(i+1)\epsilon}{s}\big)\) examples.
  This step is crucial as only in the neighbourhood of \(X\), parameters will
  have a high impact and we thus limit our search for parameters to this area.
  Over all \(s\) steps, this additional error sums up to \(\epsilon\).
Now, we non-deterministically guess a subset
  \(Y\subseteq X\) of size at most \(\ell^*\).
  Unrolling this non-deterministic guess adds a factor that depends only on
  the parameters of the problem and can thus be performed by an fpt algorithm.
  Next, we apply \cref{lem:vitali} and obtain \(Z\).
  By saving Splitter's answers in the splitter game to the picks \(z_j\in Z\)
  in the parameters \(w_j\), we obtain the vector \(\hat w^i\).
  With this vector \(\hat w^i\), we can now apply \cref{lem:projection} and compute
  the next graph \(G^{i+1}\) and the next training sequence \(\Lambda^{i+1}\)
  to proceed to Step \(i+1\).

  In Step \(s\), we know that the tuple \(\hat w^{s}\) discriminates all but
  \(m(\epsilon^*+\epsilon)\) examples from  \(\Lambda^s\) in \(G^s\).
  By concatenating all \(\hat w^{i}\), we obtain a tuple \(\bar w\) of length
  \(\ell = \ell^* \cdot s\) that discriminates all but \(m(\epsilon^*+\epsilon)\)
  examples in the original graph \(G\).
  We finish the computation by testing all possible formulas of quantifier rank
  \(q\) with \(k+\ell\) free variables and are guaranteed to find at least one
  \(\phi\) with \(\err_{\Lambda}(\phi,\bar w) \leq \epsilon^*+\epsilon\).

  Overall, the problem \RFOERM{L,Q} is in \FPT{}
  since there is an fpt algorithm for model checking first-order formulas
  on nowhere dense graphs and the bound on the number of such formulas we need to check
  depends only on \(k,\ell,\text{ and }q\).
  All intermediate steps can also be performed by fpt algorithms.
\end{proof}

 \section{Conclusion}

We have studied the parameterized complexity of learning first-order
queries from examples. The specific problem we focussed on is
\emph{empirical risk minimisation}: find a query consistent with a
set of positive or negative examples, or if no such query exists,
minimise the error. While the main reason for looking at this problem
is that it is essentially equivalent to PAC learning, we believe that
the problem is natural and interesting in its own right, en par with
other algorithmic problems typically studied for queries: evaluation
(or model-checking), enumeration, or counting. The relevance of the
problem goes beyond mere ``learning queries by example''. In database
systems, a similar problem arises for example in schema mapping \cite{catdalkol13};
the problem of finding formal specifications consistent with examples
arises in other fields as well, for example in formal verification (e.g.~\cite{lodmadnei16,garneimadrot16}).

Technically, the ERM problem poses new
challenges. In this work, they became most obvious in the hardness proof,
which required novel ideas quite distinct to those used in similar
results. It is remarkable that the hardness result even holds
for an approximate version of the problem, because hardness of
approximation results are rare in the world of parameterized
complexity theory.

Nowhere dense classes seem like a natural limit for the tractability
of the learning problem, at least for graph classes that are closed
under taking subgraphs. Adler and Adler \cite{adladl14} proved that
nowhere dense classes are the largest family of classes closed under
taking subgraphs for which first-order queries have bounded VC
dimension and hence for which we have a uniform information theoretic
PAC learning result. Furthermore, Grohe, Kreutzert and Siebertz
\cite{grokresie17} showed that nowhere dense classes are the largest
family of classes closed under taking subgraphs for which first-order
model-checking is fixed-parameter tractable. And indeed, our proof
shows that the learning problem can be tightly linked with the game
characterisation of nowhere-denseness. However, the result holds only
for the approximation version \RFOERM{L,Q} of the problem where we are
allowed to increase the hyperparameters $\ell$ and $q$. It remains a
challenging open problem whether \FOERM\ is also fixed-parameter
tractable on nowhere dense classes.

Our hardness result yields a reduction from the model-checking problem
\FOmc\ to \FOERM\ and \RFOERM{L,Q}. This reduction relativises to
graph classes satisfying mild closure conditions.
It remains open whether the problems are actually equivalent, that is,
whether there is a reduction
from \FOERM\ or \RFOERM{L,Q}\ to \FOmc\ on all classes satisfying
reasonable closure conditions.

While we know that first-order definable concepts admit no
sublinear-time learning algorithms on nowhere dense graph classes, it might
be possible to obtain such algorithms
after a polynomial-time preprocessing phase (similar to the results of
\cite{GroheLoedingRitzert_MSO,GrienenbergerRitzert_Trees} for monadic second-order logic on strings and trees).

Finally, it would be interesting to extend our results to richer logics
that implement more features of commonly used relational database systems
such as the extensions of first-order logic with counting or weight
aggregation. Similarly, it would be interesting to prove that concepts
definable in monadic second-order logic over structures of bounded
tree width or bounded clique width can be learned by fixed-parameter
tractable algorithms.

\bibliography{references}

\appendix
\newpage
\section{Appendix}

\subsection{Proof of \cref{lem:constant-l}}

\constantL*{}

\begin{algorithm}
  \caption{\FOERM{} learning algorithm for constant \(\ell\)}
  \label{algorithm:constant-l}
  \begin{algorithmic}[1]
    \REQUIRE{Training sequence \(\Lambda \in \big(V(G)^k \times \{0,1\}\big)^m\)}
    \STATE{\(\textit{min\_errors} \gets m+1\)}
    \FORALL{\(\phi' \in \Phi'\)}
    \FORALL{\(\bar{w} \in \big(V(G)\big)^{\ell}\)}
\STATE{\textit{errors} \( \gets 0\)}
    \FORALL{\((\bar{v}, \lambda) \in \Lambda\)}
    \STATE{Let \(G'\) be the graph with \(V(G')=V(G)\),
      \(R_i(G') = \{v_i\}\) for \(i \leq k\),
      \(S_j(G') = \{w_j\}\) for \(j \leq \ell\),
and \(R(G') = R(G)\) for other relations}
    \IF{(\(\lambda = 0\) and \(G' \models \phi'\)) or
      (\(\lambda = 1\) and \(G' \not\models \phi'\))}
    \STATE{\(\textit{errors} \leftarrow \textit{errors} + 1\)}
    \ENDIF
    \ENDFOR
    \IF{\(\textit{errors} < \textit{min\_errors}\)}
    \STATE{\(\phi_{\min}' \leftarrow \phi'\), \(\bar{w}_{\min} \leftarrow \bar{w}\)}
    \ENDIF
    \ENDFOR
    \ENDFOR
    \RETURN{\(\big(\phi_{\min}(\bar{x}\smid\bar{y}),\,\bar{w}_{\min}\big)\),
      where \(\phi_{\min}\) is obtained from \(\phi_{\min}'\) by replacing
      all occurrences of \(R_i x\) with \(x=x_i\) and all occurrences of
      \(S_i x\) are with \(x=y_i\)}
  \end{algorithmic}
\end{algorithm}

\begin{proof}
We start the proof of the learnability for the setting that
  \(\ell\) is constant by adding
  \(k+\ell\) unary relations (colours) to \(\tau\):
  \(\tau' \coloneqq \tau \uplus \{R_1, \dots, R_k, S_1, \dots, S_{\ell}\}\).
  Then we have a finite set \(\Phi' \subseteq \FO[\tau', q]\)
  of formulas in Gaifman normal form (here: sentences in Gaifman normal form)
  with quantifier rank at most \(q\).
  The algorithm for finding a hypothesis is given in pseudocode
  in \cref{algorithm:constant-l}.
  We search for a sentence that minimises the training error in the graph \(G'\)
  over the extended alphabet \(\tau'\).
  This sentence is then translated back into a formula with \(k+\ell\)
  free variables by substituting all occurrences of the additional relations.
  The translated formula is then returned by the algorithm,
  solving the problem \FOERM{} and its relaxation \RFOERM{L,Q}.

  Let \(n = \abs{V(G)}\).
  Using the assumption that \(k,\ell \leq n\),
  \cref{algorithm:constant-l} runs in time
  \[
    |\Phi'| \cdot n^{\ell}\! \cdot m \cdot
      \big(k + \ell\! + n + g(q\!+ |\tau'|) \cdot n^c\big)
    \in \bigO\big(f(k+q\!+\abs{\tau}) \cdot m \cdot n^{\ell\!+ c + 1}\big).
  \]
  for some functions \(f\) and \(g\).
  In this running time, the first three factors correspond to the for-loops
  of the algorithm and the bracketed part comes from model checking.
\end{proof}

\subsection{Proof of \cref{lem:k-one}}

\kIsOne*{}

  \begin{algorithm}
  \caption{\FOERM{} learning algorithm for \(k=1\)}
  \label{algorithm:k-one}
  \begin{algorithmic}[1]
    \REQUIRE{Training sequence \(\Lambda \in \big(V(G)^k \times \{0,1\}\big)^m\)}
    \FORALL{\(\phi \in \Phi'\)}
    \STATE{\(\textit{consistent} \leftarrow \TRUE\)}
    \FOR{\(i=1\) to \(\ell\)}
    \STATE{\(\phi_i(x, y_{i+1}, \dots, y_{\ell})
      \coloneqq \exists y_1 \dots \exists y_i \big( \bigwedge_{j=1}^i S_j y_j\)
      \hfill\\ \hspace{4.5cm} \(\land \phi(x, y_1, \dots, y_{\ell})\big)\)}
    \STATE{\(\textit{found} \leftarrow \FALSE\)}
    \FORALL{\(u \in V(G)\)}
    \STATE{Let \(G'\) be the graph with \(V(G')=V(G)\),
      \(S_j(G') = \{w_j\}\) for all \(j < i\),
      \(S_i(G') = \{u\}\),
      \(P_+(G') = \Lambda^+\),
      \(P_-(G') = \Lambda^-\),
      and \(R(G') = R(G)\) for other relations}
    \IF{\(G' \models \exists y_{i+1} \dots \exists y_{\ell}\! \forall x
      \Big(\big(P_+x \to \phi_i(x, y_{i+1}, \dots, y_{\ell}\!)\big) \land\)
      \hfill\\\hfill
      \(\big(P_-x \to \neg \phi_i(x, y_{i+1}, \dots, y_{\ell}\!)\big)\Big)\)}
    \STATE{\(w_i \leftarrow u\)}
    \STATE{\(\textit{found} \leftarrow \TRUE\)}
    \BREAK
    \ENDIF
    \ENDFOR
    \IF{not \(\textit{found}\)}
    \STATE{\(\textit{consistent} \leftarrow \FALSE\)}
    \BREAK
    \ENDIF
    \ENDFOR
    \IF{\(\textit{consistent}\)}
    \RETURN{\(\big(\phi(x\smid\bar{y}),\,\bar{w}\big)\)}
    \ENDIF
    \ENDFOR
    \REJECT
  \end{algorithmic}
\end{algorithm}

\begin{proof}
  Let \(\tau' \coloneqq \tau \uplus \{S_1, \dots, S_{\ell}, P_+, P_-\}\).
Again, similar to the setting in \cref{lem:constant-l},
  we have a finite set \(\Phi' \subseteq \FO[\tau',q]\)
  of formulas in Gaifman normal form with quantifier rank at most \(q\)
  with \(\ell+1\) free variables.
  The algorithm uses a relation \(P_+\) for all positive and a relation \(P_-\)
  for all negative examples.
  Since \(k=1\), those relations are both unary and the resulting graph is still
  contained in \(\CC\).
  Using those unary relations for the examples, it is possible to test whether
  a given prefix of parameters can be extended to a consistent one.
  The algorithm thus starts by testing for every node \(u\in V(G)\)
  whether it can be extended to a consistent parameter setting.
  If it finds such a node \(u\), it fixes \(w_1=u\) and proceeds with the search
  for \(w_2\).
  Let \(n = \abs{V(G)}\).
  After at most \(\ell \cdot n\) model-checking steps,
  the algorithm has discovered a consistent parameter setting if there is one.
  As in \cref{lem:constant-l}, we use model checking for sentences only
  and expand the background structure to contain additional unary relations
  (colours) for the free variables of the formula.
  In our case, the free variables of the formulas we evaluate are exactly the
  parameter prefixes. The algorithm is given in pseudocode in
  \cref{algorithm:k-one} and runs in time
  \[
    |\Phi'| \cdot \ell\! \cdot n \cdot
      \big(n + m + \ell\! + g(q\!+\ell\!+|\tau'|) \cdot n^c\big)
    \in \bigO\big(f(\ell\!+q\!+\abs{\tau}) \cdot m \cdot n^{c+1}\big).
  \]
  for some functions \(f\) and \(g\).
  In this running time, the first three factors correspond to the for-loops
  of the algorithm and the bracketed part comes from model checking.
\end{proof}

\subsection{Proof of \cref{lem:projection}}

\projection*{}

\begin{proof}[Proof of \cref{lem:projection}]
  We let \(G^{i+1}\) be the graph obtained from the induced
  \(R'\)-neighbourhood structure \(\Neighb{R'}{G^i}{Z}\) as follows.
We start with the construction;
  the roles of each of the steps will become clear
  while proving the statements of the lemma.
  \begin{enumerate}
    \item Expand the graph by fresh colours
    \(D_{j,d}\) for \(j \in [\ell']\) and \(d \in \{0, \dots,(k+2)(2r+1)\}\).
    We let \(D_{j,d}(G^{i+1})\) be the set of all \(v\)
    such that \(\dist_{G^i}(v,y_j)=d\).
    \item Expand the graph by fresh colours \(C_{j}\) for \(j \in [\ell'']\).
    Let \(C_{j}(G^{i+1}) \coloneqq \neighb{1}{G^i}{w_j}\).
    \item Delete all edges incident with the \(w_j\).
    Moreover, we add fresh colours \(B_j\) for \(j \in [\ell'']\)
    and let \(B_j(G^{i+1}) \coloneqq \{w_j\}\).
    \item For each non-empty set \(I \subset [k]\)
    and each \(\abs{I}\)-variable \(q^*\)-type
    \(\theta \in \Types{\tau}{\abs{I}}{q^*}\),
    we add an isolated vertex \(t_{I,\theta}\)
    and a fresh colour \(A_{I,\theta}\)
    and let \(A_{I,\theta}(G^{i+1}) = \{t_{I,\theta}\}\).
  \end{enumerate}

Thus, structurally, \(G^{i+1}\) consists of the neighbourhoods
  \(G_j^{i+1} \coloneqq G^i \big[\neighb{R'}{G^i}{z_j} \setminus \{w_j\}\big]\)
  for \(j \in [\ell'']\)
  and isolated vertices \(w_1, \dots, w_{\ell''}\)
  and \(t_{I,\theta}\) for all non-empty \(I \subset [k]\)
  and \(\theta \in \Types{\tau}{\abs{I}}{q^*}\).
  Hence, Statements \eqref{lem:proj:nodes} and \eqref{lem:proj:edges} of
  \cref{lem:projection} hold by the construction of \(G^{i+1}\).
  Note that the neighbourhoods \(G^{i+1}_j\) and \(G^{i+1}_{j'}\)
  are disconnected for each \(j\neq j'\) by the construction of \(Z\).

  Splitter's winning strategy on \(G^i\) is still valid on \(G^{i+1}\),
  but one of the steps of the splitter game
  (removing a vertex Spoiler chose and continuing in the neighbourhood of a vertex
  Connector chose) has already been performed
  in each of the neighbourhoods \(G_j^{i+1}\)
  in the construction of \(G^{i+1}\).
  Hence, Statement~\eqref{lem:proj:winning} follows.

Next, we define the training sequence
  \(\Lambda^{i+1} \in \big(V(G^{i+1})^k \times \{0,1\}\big)^{m_{i+1}}\)
  for the next round.
  For this step, we will need the isolated vertices introduced in Step 4
  in the construction of \(G^{i+1}\).
  We only need to consider examples in \(\Gamma^i\),
  which are involved in a conflict;
  all other examples can be correctly classified without the use of parameters.
For each tuple \(\bar{v} = (v_1, \dots, v_k) \in \Gamma_i\)
  with \(\bar{v} \cap \neighb{6r+3}{G^i}{Y} \neq \emptyset\),
  we define a tuple \(\bar{v}' = (v_1', \dots, v_k') \in V(G^{i+1})\)
  for which we add \((\bar{v}', \lambda)\) to \(\Lambda^{i+1}\)
  if \((\bar{v}, \lambda) \in \Lambda^i\).
  For that, we consider the graph \(H_{\bar{v}}\)
  with vertex set \(V(H_{\bar{v}}) \coloneqq [k]\)
  and edges \(\{a,b\}\) for all \(a,b \in [k]\) with
  \(1 \leq \dist_{G^i}(v_a,v_b) \leq 2r+1\).
  Let \(I_1, \dots, I_p\) be the vertex sets
  of the connected components of \(H_{\bar{v}}\).
  For each component \(I_j\), we proceed as follows:
  if there is some \(j_0 \in I_j\) such that
  \(v_{j_0} \in \neighb{6r+3}{G^i}{Y}\),
  we let \(v_a' \coloneqq v_a\) for all \(a \in I_j\).
  Note that for all \(a \in I_j\) we have
  \(\dist_{G^i}(v_{j_0}, v_a) \leq (k-1)(2r+1)\) and hence,
  \(\dist_{G^i}(Y, v_a) \leq 6r+3 + (k-1)(2r+1) = (k+2)(2r+1)\).
  Thus, \(v_a' \in \neighb{(k+2)(2r+1)}{G^i}{Y} \subseteq
  \neighb{R'}{G^i}{Z} \subseteq V(G^{i+1})\).
  This is the point that determines the choice of \(R\)
  of the splitter game in the very beginning of the proof.
  Otherwise, if \(v_a \not\in \neighb{6r+3}{G^i}{Y}\) for all \(a \in I_j\),
  we consider the restriction \(\bar{v}|_{I_j}\) of \(\bar{v}\)
  to the indices in \(I_j\) and let
  \(\theta \coloneqq \ltype{q^*}{r}{G^i}{\bar{v}|_{I_j}}\).
  In this case, we let \(v_a' \coloneqq t_{I_j,\theta}\) for all \(a \in I_j\).
  Note that for some examples this means that they only consist of isolated
  vertices and thus will never be correctly classified by our algorithm.
Observe that two tuples \(\bar{v}_1', \bar{v}_2'\)
  that appear in examples in \(\Lambda^{i+1}\)
  can only have the same type in \(G^{i+1}\)
  if their counterparts \(\bar{v}_1, \bar{v}_2\) from \(\Lambda^i\)
  have the same type in \(G^i\).
  Hence, we do not create any new conflicts by the construction.

  Let \(A = \Neighb{R'}{G^i}{Z}\) and \(A'\) the modified version
  of this neighbourhood in \(G^{i+1}\).
  It is easy to interpret \(A\) in \(G^{i+1}\) using the information encoded
  in the fresh colours added in Step 1-3 of the construction of \(G^{i+1}\).
  Using the parameters \(w_1,\dots,w_{\ell''}\), we can also interpret the
  modified neighbourhood \(A'\) in \(G^i\). We use this connection to prove
  Statement \eqref{lem:proj:down} and \eqref{lem:proj:up}.
  Note that for encoding the distance information, the increased quantifier
  rank \(q\) might be necessary which is why we chose \(q=q^*+\log R\).

Assume there is a tuple \(\bar w^i \in V(G^i)^{\ell^*}\) that discriminates
  all but \(m\epsilon'\) examples in \(G^i\).
  By \cref{lem:existenceOfDiscriminating}, there is a tuple
  \(\bar u^i \in (\neighb{4r+2}{G^i}{Z})^{\ell^*}\) that discriminates all but
  \(m(\epsilon'+\frac{\epsilon}{s})\) examples in \(G^i\).
  Then \(\bar u^i\) also discriminates all but
  \(m(\epsilon'+\frac{\epsilon}{s})\) examples in \(G^{i+1}\) due to the way we
  projected the examples and the fact that we can emulate \(A\) in \(G^{i+1}\).
  This shows Statement~\eqref{lem:proj:down}.

If we can distinguish \(\bar{v}_1', \bar{v}_2'\) in \(G^{i+1}\)
  using parameters \(\bar{w}^{i+1}\),
  then we can distinguish \(\bar{v}_1, \bar{v}_2\) in \(G^i\)
  using \(\bar{w}^{i+1}\) and the chosen parameters \(w_1, \dots, w_{\ell''}\).
  This holds as we can interpret \(A'\) in \(G^i\).
  Thus, if \(\bar{w}^{i+1}\) discriminates all but
  \(c\) examples in \(\Lambda^{i+1}\),
  then \(\hat{w}^i \bar{w}^{i+1}\) discriminates all but
  \(c\) examples in \(\Lambda^i\).
  This shows Statement \eqref{lem:proj:up}.

To finish the proof,
  we need to show that all steps can be performed by an fpt algorithm.
  In the first three steps,
  the number of new colours only depends on the parameters \(k\), \(\ell^*\),
  and \(q^*\).
  In step 4, the number of new colours only depends on \(k\), \(q^*\)
  and the size of the signature of \(G^i\).
  Since the number of computed graphs is bounded by \(s\)
  (and therefore in terms of \(q^*\), \(k\), and \(\ell^*\)),
  the total number of fresh colours only depends on the parameters.
  Furthermore, the computations are linear in the size of \(G^i\).
  For the projection of the examples,
  the computation of the types can be performed by an fpt algorithm
  because \(G^i\) is nowhere dense.
  All other computations for the projection
  can be done by an fpt algorithm as well.
\end{proof}

\end{document}